\documentclass[aps,pra,10pt,letterpaper,twocolumn,tightenlines,superscriptaddress,notitlepage,longbibliography,nofootinbib]{revtex4-2}

\usepackage[utf8]{inputenc}
\usepackage[english]{babel}
\usepackage[T1]{fontenc}
\usepackage{amsmath}
\usepackage{stmaryrd}

\makeatletter
\newcommand{\pushright}[1]{\ifmeasuring@#1\else\omit\hfill$\displaystyle#1$\fi\ignorespaces}
\newcommand{\pushleft}[1]{\ifmeasuring@#1\else\omit$\displaystyle#1$\hfill\fi\ignorespaces}
\makeatother

\makeatletter 
    
\renewcommand\onecolumngrid{
\do@columngrid{one}{\@ne}%
\def\set@footnotewidth{\onecolumngrid}
\def\footnoterule{\kern-6pt\hrule width 1.5in\kern6pt}%
}

\renewcommand\twocolumngrid{
        \def\footnoterule{
        \dimen@\skip\footins\divide\dimen@\thr@@
        \kern-\dimen@\hrule width.5in\kern\dimen@}
        \do@columngrid{mlt}{\tw@}
}%

\makeatother    

\allowdisplaybreaks

\usepackage{tikz}
\usepackage{lipsum}
\usepackage{amssymb}
\usepackage{color}
\usepackage{mathrsfs}
\usepackage{amsbsy}
\usepackage{amsthm}
\usepackage{array}
\newcolumntype{M}[1]{>{\centering\arraybackslash}m{#1}}
\newcolumntype{N}{@{}m{0pt}@{}}
\usepackage{balance}
\usepackage{thmtools,thm-restate}
\usepackage{graphicx}
\usepackage{tikz}

\usepackage{comment}

\usepackage{bbm}
\usepackage{bm}
\usepackage{epsfig}
\usepackage{xfrac}
\usepackage{xcolor}
\usepackage{enumerate}
\usepackage[shortlabels]{enumitem}
\usepackage{graphicx}
\usepackage{bm}
\usepackage{hyperref}
\usepackage{booktabs}
\usepackage{makecell}
\usepackage[linesnumbered,ruled,vlined]{algorithm2e}

\usepackage{subfigure}
\hypersetup{
colorlinks=true,
linkcolor=blue,
filecolor=blue,
citecolor=blue,  
urlcolor=blue,
}

\newtheorem{definition}{Definition}

\newtheorem{theorem}{Theorem}
\newtheorem{theorem-main}{Theorem}
\newtheorem{corollary}{Corollary}
\newtheorem{lemma}{Lemma}

\newcommand{\ketbra}[2]{|#1\rangle\langle #2|}

\newcommand{\projP}{\mathsf{P}}
\newcommand{\htheta}{\hat{\theta}}
\newcommand{\hrho}{\hat{\rho}}

\newcommand{\tr}{{\mathrm{Tr}}}
\DeclareMathOperator{\Tr}{Tr}

\usepackage{stackengine}

\usepackage{braket}
\renewcommand{\v}[1]{\ensuremath{\mathbf{#1}}} 
\newcommand{\gv}[1]{\ensuremath{\text{\boldmath$ #1 $}}}
\newcommand{\abs}[1]{\left| #1 \right|}
\newcommand{\norm}[1]{\left\| #1 \right\|} 

\newcommand{\trace}{\mathrm{Tr}}

\newcommand{\tq}{{\tilde{q}}}
\newcommand{\tj}{{\tilde{j}}}
\newcommand{\tk}{{\tilde{k}}}

\newcommand{\mS}{{\mathcal{S}}}

\newcommand{\mH}{{\mathcal{H}}}

\newcommand{\mU}{{\mathcal{U}}}

\newcommand{\mM}{{\mathcal{M}}}

\newcommand{\tJ}{{\widetilde{J}}}
\newcommand{\tM}{{\tilde{M}}}
\newcommand{\tL}{{\widetilde L}}

\newcommand{\tK}{{\widetilde K}}

\newcommand{\mE}{{\mathcal{E}}}

\newcommand{\id}{{\mathbbm{1}}}

\newcommand{\vv}{{\gv{v}}}

\newcommand{\vX}{{\v{X}}}

\newcommand{\bR}{{\mathbb{R}}}
\newcommand{\bZ}{{\mathbb{Z}}}

\newcommand{\bE}{{\mathbb{E}}}

\renewcommand{\Re}{{\mathrm{Re}}}
\renewcommand{\Im}{{\mathrm{Im}}}

\newcommand{\opt}{{\mathrm{opt}}}
\renewcommand{\epsilon}{\varepsilon}
\newcommand{\appropto}{\mathrel{\vcenter{
  \offinterlineskip\halign{\hfil$##$\cr
    \propto\cr\noalign{\kern2pt}\sim\cr\noalign{\kern-2pt}}}}}

\definecolor{fluorescentpink}{rgb}{1.0, 0.08, 0.58}

\let\baraccent=\= 
\renewcommand{\=}[1]{\stackrel{#1}{=}} 

\newcommand{\thmref}[1]{\hyperref[#1]{Theorem~\ref{#1}}}
\newcommand{\lemmaref}[1]{\hyperref[#1]{Lemma~\ref{#1}}}
\newcommand{\propref}[1]{\hyperref[#1]{Proposition~\ref{#1}}}
\newcommand{\corollaryref}[1]{\hyperref[#1]{Corollary~\ref{#1}}}
\newcommand{\defref}[1]{\hyperref[#1]{Definition~\ref{#1}}}
\newcommand{\figref}[1]{\hyperref[#1]{Fig.~\ref{#1}}}
\newcommand{\tabref}[1]{\hyperref[#1]{Table~\ref{#1}}}
\newcommand{\figaref}[1]{\hyperref[#1]{Fig.~\ref{#1}(a)}}
\newcommand{\figbref}[1]{\hyperref[#1]{Fig.~\ref{#1}(b)}}
\newcommand{\figcref}[1]{\hyperref[#1]{Fig.~\ref{#1}(c)}}
\newcommand{\figdref}[1]{\hyperref[#1]{Fig.~\ref{#1}(d)}}
\newcommand{\figeref}[1]{\hyperref[#1]{Fig.~\ref{#1}(e)}}
\newcommand{\figfref}[1]{\hyperref[#1]{Fig.~\ref{#1}(f)}}
\renewcommand{\eqref}[1]{\hyperref[#1]{Eq.~(\ref{#1})}}
\newcommand{\secref}[1]{\hyperref[#1]{Sec.~\ref{#1}}}
\newcommand{\eqsref}[2]{\hyperref[#1]{Eqs.~(\ref{#1})-(\ref{#2})}}
\newcommand{\appref}[1]{\hyperref[#1]{Appx.~\ref{#1}}}

\begin{document}

\title{Randomized measurements for multi-parameter quantum metrology}

\author{Sisi Zhou}\email{sisi.zhou26@gmail.com}
\affiliation{Perimeter Institute for Theoretical Physics, Waterloo, Ontario N2L 2Y5, Canada}
\affiliation{Department of Physics and Astronomy, Department of Applied Mathematics and Institute for Quantum Computing, University of Waterloo, Ontario N2L 2Y5, Canada}

\author{Senrui Chen}\email{csenrui@gmail.com}
\affiliation{Pritzker School of Molecular Engineering, The University of Chicago, Chicago, Illinois 60637, USA}
\affiliation{Institute for
Quantum Information and Matter, California Institute of Technology, Pasadena, CA 91125, USA}

\date{\today}

\begin{abstract}
The optimal quantum measurements for estimating different unknown parameters in a parameterized quantum state are usually incompatible with each other. Traditional approaches to addressing the measurement incompatibility issue, such as the Holevo Cram\'{e}r--Rao bound, 
suffer from multiple difficulties towards practical applicability, as the optimal measurement strategies are usually state-dependent, difficult to implement and also take complex analyses to determine.
Here we study randomized measurements as a new approach for multi-parameter quantum metrology. We show quantum measurements on single copies of quantum states given by $3$-designs perform near-optimally when estimating an arbitrary number of parameters in pure states and more generally, {approximately low-rank well-conditioned states}, whose metrological information is largely concentrated in a low-dimensional subspace. The near-optimality is also shown in estimating the maximal number of parameters for three types of mixed states that are well-conditioned on their supports. Examples of fidelity estimation and Hamiltonian estimation are explicitly provided to demonstrate the power and limitation of randomized measurements in multi-parameter quantum metrology. 
\end{abstract}

\maketitle

\begin{figure*}[t]
    \centering
    \includegraphics[width=1.4\columnwidth]{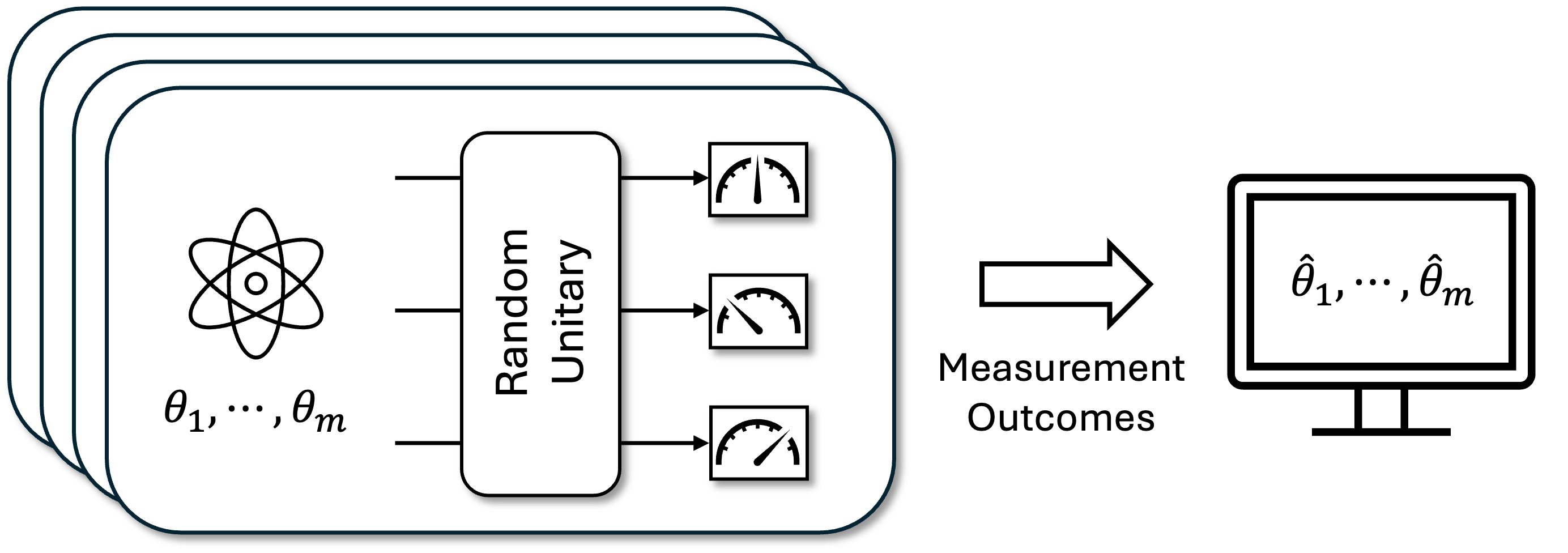}
    \caption{Schematics of randomized measurements for multi-parameter quantum metrology. 
    The measurement is taken on \emph{individual} copies of the same parameterized state $\rho_\theta:=\rho_{\theta_1,\cdots,\theta_m}$. The goal is to obtain estimates $\hat\theta=(\hat\theta_1,\cdots,\hat\theta_m)$ with the highest precision as the number of copies approaches infinity.
    In all of our theorems showing the near-optimality of randomized measurements, it suffices to take the measurement to be a $3$-design. }
\end{figure*}
\begin{table*}[t]
    \centering
    \begin{tabular}{l l l}
            \toprule
              & Traditional approaches & $3$-design measurements \\
             \midrule
             \makecell[l]{Computational complexity \\ (for finding measurements)} & Usually at least polynomial in system dimension & None (state-independent) \\
             Implementation complexity & Usually unknown (polynomial in system dimension) & Polynomial in qubit number 
             \\
             \bottomrule
        \end{tabular}   
    \caption{Comparison of traditional (individual) measurement protocols and our randomized measurement protocol for multi-parameter quantum metrology. See \appref{app:review} for detailed explanations and extended discussions on collective measurements. }
    \label{tab:comparison}
\end{table*}
\begin{table*}[t]
    \centering
    \begin{tabular}{l l l}
            \toprule
             Families of Parameterized States & Randomized Measurements & Examples\\
             \midrule
             Pure states (\secref{sec:pure})& Near-optimal (\thmref{thm:pure}) & \secref{sec:fid_pure}, \secref{sec:ham_noiseless}\\
             Approximately low-rank well-conditioned states (\secref{sec:low-rank}) & Near-optimal (\thmref{thm:unitary}, \thmref{thm:low-rank}) & \secref{sec:ham_noisy}\\
             Full-parameter rank-$r$ well-conditioned states (\secref{sec:full}) & Weakly near-optimal (\thmref{thm:full}) & - \\
             General mixed states & Not always (weakly) near-optimal & \secref{sec:informal}, \secref{sec:no-go}\\
             \bottomrule
        \end{tabular}   
    \caption{Summary of main results.}
    \label{tab:summary}
\end{table*}

\section{Introduction}

Quantum metrology is the study of parameter estimation in quantum systems~\cite{giovannetti2011advances,degen2017quantum,pezze2018quantum,pirandola2018advances}. It has extensive applications in optical interferometry~\cite{caves1981quantum,yurke19862,ligo2011gravitational,ligo2013enhanced}, frequency detection~\cite{wineland1992spin,bollinger1996optimal,leibfried2004toward,taylor2008high,zhou2020quantum,rosenband2008frequency,appel2009mesoscopic,ludlow2015optical,kaubruegger2021quantum,marciniak2022optimal} and quantum imaging~\cite{le2013optical,lemos2014quantum,tsang2016quantum,abobeih2019atomic}, etc. In single-parameter estimation, the quantum Cram\'{e}r--Rao bound (QCRB) characterizes the ultimate precision limit of estimating one parameter in quantum states through an information-theoretic quantity called the quantum Fisher information (QFI)~\cite{helstrom1967minimum,helstrom1968minimum,helstrom1976quantum,holevo2011probabilistic,braunstein1994statistical,barndorff2000fisher,paris2009quantum}, which is the classical Fisher information (CFI) maximized over all possible quantum measurements on quantum states. In multi-parameter estimation, however, the quantum Fisher information matrix (QFIM) and the corresponding quantum Cram\'{e}r--Rao bound is not always attainable because the optimal quantum measurements that correspond to different parameters may not commute and thus cannot be simultaneously implemented. Trade-offs among the estimation precisions for different parameters must be taken into consideration~\cite{gill2000state,matsumoto2002new,nagaoka2005new,conlon2021efficient,ragy2016compatibility,demkowicz2020multi,szczykulska2016multi,liu2020quantum,sidhu2020geometric,kahn2009local,yamagata2013quantum,yang2019attaining,humphreys2013quantum,vidrighin2014joint,suzuki2016explicit,chen2017maximal,pezze2017optimal,li2016fisher,zhu2018universally,vargas2024near,yang2019optimal,carollo2019quantumness,albarelli2020perspective,suzuki2020quantum,hou2020minimal,belliardo2021incompatibility,candeloro2021properties,lu2021incorporating,chen2022information,chen2022incompatibility,yung2024comparison}. 

One popular solution of the measurement incompatibility issue in multi-parameter metrology is to minimize the weighted sum of estimation variances for a given cost matrix, instead of aiming at optimizing the estimation variances for all parameters. Holevo's version of the quantum Cram\'{e}r--Rao bound (HCRB)~\cite{holevo2011probabilistic,kahn2009local,yamagata2013quantum,yang2019attaining}, which provides a tight lower bound for the weighted sum, was extensively researched and applied~\cite{ragy2016compatibility,demkowicz2020multi,albarelli2019evaluating,gorecki2020optimal,tsang2020quantum,sidhu2021tight,hayashi2023tight,gardner2024achieving}. It is numerically solvable through semidefinite programs~\cite{albarelli2019evaluating} and analytically solvable when the number of parameters is two~\cite{sidhu2021tight}. 
The HCRB is attainable for generic quantum states using collective measurements on infinitely many copies of quantum states~\cite{kahn2009local,yamagata2013quantum,yang2019attaining}. When restricted to individual measurements on single copies of states, the HCRB is attainable for pure states~\cite{matsumoto2002new}. Individual measurements also have known solutions for some specific families of states, such as single-qubit states, via other numerically solvable bounds (e.g., the Nagaoka--Hayashi bound~\cite{nagaoka2005new,conlon2021efficient} and the Gill--Massar bound~\cite{gill2000state}). 
Another line of research in multi-parameter estimation focuses on local state tomography, where optimal measurements for uniformly estimating all parameters in pure states called Fisher-symmetric measurements were found and studied~\cite{li2016fisher,zhu2018universally,vargas2024near}.

Despite the success of previous efforts in addressing multi-parameter estimation problems in quantum metrology, several obstacles exist against its practical usefulness, in particular when the \emph{system dimension} is large. The first obstacle is the \emph{implementation complexity}. Collective measurements on a large number of states are required in general, which are almost experimentally infeasible. In situations where individual measurements are sufficient~\cite{matsumoto2002new,nagaoka2005new,conlon2021efficient,li2016fisher,zhu2018universally,vargas2024near}, the implementation can still be challenging when the measurement basis is complex. The second obstacle is the \emph{computational resource} required for identifying the optimal measurements. The numerical computation of the HCRB scales polynomially with respect to the system dimension, which can become quickly intractable e.g., for multi-qubit systems as the number of qubits increases. Note that searching for optimal individual measurements that minimize the weighted sum of estimation variance is an NP-hard problem that involves minimization over a separable cone~\cite{hayashi2023tight} which is even more challenging. Finally, the optimal quantum measurements obtained from traditional approaches usually depend on the quantum state, meaning they need to be updated adaptively as the prior knowledge of the state is refined during the experiment. The optimal measurements can also depend on the cost matrix, limiting its application scenarios, e.g., in situations where the cost matrix is revealed only in the post-processing phase after the measurement~\cite{elben2023randomized}.

Here we propose randomized measurement protocols for multi-parameter metrology, where the above mentioned obstacles are alleviated (see \tabref{tab:comparison}). Random unitaries and measurements~\cite{watrous2018theory} are widely applied in quantum device benchmarking~\cite{knill2008randomized,dankert2009exact,emerson2005scalable}, learning and tomography~\cite{elben2023randomized,huang2020predicting,brydges2019probing}, etc., but were rarely explored in multi-parameter quantum metrology. Below we mostly focus on randomized measurements on individual copies of quantum states whose measurement basis forms a (complex projective) $3$-design~\cite{renes2004symmetric,ambainis2007quantum}, which can be realized through applying a unitary $3$-design~\cite{dankert2009exact,roy2009unitary,gross2007evenly} prior to a projection onto the computational basis. The implementation of $3$-designs was thoroughly investigated in previous literature.  In multi-qubit systems, random Clifford circuits form a unitary $3$-design~\cite{kueng2015qubit,webb2015clifford,zhu2017multiqubit} and can be efficiently implemented through linear-depth quantum circuits in 1D architecture~\cite{aaronson2004improved,bravyi2021hadamard} and log-depth circuits in all-to-all connected architecture with ancilla~\cite{moore2001parallel,jiang2020optimal}. For general qudit systems, approximate implementations~\cite{brandao2016local,haferkamp2023efficient,harrow2023approximate,schuster2024random} and explicit constructions~\cite{bajnok1998constructions,bannai2019explicit,bannai2022explicit} were also proposed. 

In this work, we demonstrate the optimality of randomized measurements for different families of quantum states (see \tabref{tab:summary}). We first prove the near-optimality of randomized measurements for multi-parameter pure state estimation. We show the classical Fisher information matrix (CFIM) obtained from randomized measurements attains the QFIM up to a constant factor no larger than 4 for any parameterized pure states in any finite-size systems. Next, we consider {approximately low-rank well-conditioned states}, which are states whose metrological information is largely concentrated in a low-dimensional subspace, and further demonstrate the saturation of the QFIM up to a constant factor, extending the result of pure states. For both pure states and {approximately low-rank well-conditioned states}, the near-optimality is independent of the choice of cost matrices, aligning with the ``measure first, ask questions later'' framework~\cite{elben2023randomized}. We also investigate mixed states in general, and show the near-optimality of randomized measurements for three families of quantum mixed states which have a maximal number of unknown parameters. Near-optimality is proved when the cost matrix is the QFIM, a metric that was commonly used to guarantee invariance under reparameterization and equal importance of all parameters~\cite{gill2000state,li2016fisher,zhu2018universally}. Finally, we provide examples, including fidelity estimation and Hamiltonian estimation (with and without quantum noise), demonstrating the applicability and limitation of randomized measurements for different quantum metrological tasks. 
In particular, we provide a mixed state fidelity estimation example where randomized measurements do not perform nearly optimally. 
Apart from theoretical results, we numerically implement a pure state fidelity estimation example using randomized measurements and a newly-designed near-optimal estimator, validating our theory and demonstrating a clear advantage over existing protocols (\figref{fig:numerics}).
These results highlight randomized measurements as a powerful tool for multi-parameter quantum metrology with the potential for broad applications.


\section{Informal Overview}
\label{sec:informal}

Below we first briefly introduce the setting of multi-parameter quantum metrology (a more formal introduction is deferred to \secref{sec:prelim}) and then we provide an informal overview of our results, with a focus on explaining the underlying intuition behind them. 

Consider a $d$-dimensional parameterized quantum state $\rho_{\theta}$ in Hilbert space $\mH$, where parameters are denoted by $\theta = (\theta_1,\theta_2,\ldots,\theta_m)$. $m$ is the number of parameters and $\Theta \subseteq \bR^m$ is the domain of $\theta$. An estimator $\htheta(x)$ and the corresponding positive operator-valued measurement (POVM) $M = \{M_x\}_{x}$ is a function that maps the measurement outcomes $x$ to $\Theta$. Here $M_x$ are positive semidefinite measurement operators satisfying $\sum_x M_x = \id$.  The performance of an estimator at $\theta$ is characterized by the \emph{mean squared error matrix} (MSEM)
\begin{equation}
    V(M,\htheta)_{ij} = \sum_x (\htheta(x)_i - \theta_i)(\htheta(x)_j - \theta_j) \trace(\rho_\theta M_x). 
\end{equation}
For single-parameter estimation (i.e., $m = 1$), $V$ is a scalar called the mean squared error (MSE). 

In the paradigm of local parameter estimation, we focus on analyzing the performance of \emph{locally unbiased estimators} that are unbiased at the true value of $\theta$ and in its vicinity up to the first order. It is relevant in metrological situations where the prior knowledge of $\theta$ is already given in the vicinity of its true value, or roughly obtained from a pre-estimation phase (which usually consumes a number of experiments that are negligible compared to $N$, the total number of copies of states available)~\cite{rao1973linear,kay1993fundamentals,lehmann2006theory,cox2017inference,casella2002statistical,gill2000state,yang2019attaining}. 

The Cram\'{e}r--Rao bound (CRB)~\cite{rao1973linear,kay1993fundamentals,lehmann2006theory,cox2017inference,casella2002statistical} states that for any locally unbiased estimator at $\theta$, 
\begin{equation}
\label{eq:CRB}
    V(M,\htheta) \succeq I(M)^{-1},
\end{equation}
where $A \succeq B$ denotes a relationship between two positive semidefinite matrices $A$ and $B$ where $A - B$ is positive semidefinite. Note that here we use ``$\succeq$'' instead of ``$\geq$'' to emphasize the objects under comparison are matrices, while later on when the discussion is restricted to single-parameter estimation and $I(M)$ and $J$ become scalars, we will use $\geq$ instead. 
$I(M)$ is the classical Fisher information matrix (CFIM), defined by 
\begin{equation}
    I(M)_{ij} = \sum_{x, p_x:=\trace(\rho_\theta M_x) \neq 0} \frac{1}{p_x} \frac{\partial p_x}{\partial \theta_i} \frac{\partial p_x}{\partial \theta_j}. 
\end{equation}
$I(M)$ is a function of the POVM and $\rho_\theta$, and is independent of the estimators. 
The quantum Cram\'{e}r--Rao bound (QCRB)~\cite{helstrom1967minimum,helstrom1968minimum,helstrom1976quantum,holevo2011probabilistic,braunstein1994statistical,barndorff2000fisher,paris2009quantum} further states that, for any locally unbiased estimator at $\theta$, 
\begin{equation}
\label{eq:QCRB}
    V(M,\htheta) \succeq I(M)^{-1} \succeq J^{-1},
\end{equation}
where $J$ is the quantum Fisher information matrix (QFIM) as a function of $\rho_\theta$. 
For single-parameter estimation, the quantum Fisher information (QFI) $J$ is equal to the classical Fisher information (CFI) $I(M)$ maximized over all POVMs~\cite{braunstein1994statistical}. For multi-parameter estimation, however, this does not always hold. The reason is the measurements that are optimal with respect to different parameters may not be compatible with each other, i.e., cannot be implemented simultaneously. 

Since the CFIM in general cannot attain the QFIM, we introduce the concept of \emph{near-optimality} (see details in \secref{sec:near-opt}) and investigate situations where near-optimal measurements exist. It is a property of a family of POVMs, with a corresponding state family, stating that the ratio between the CFIM and the QFIM is bounded below by a constant. 

We find in \thmref{thm:pure} that $3$-design measurements are near-optimal for pure states. 
A simple example below provides some intuition behind this result. Consider a one-qubit pure state $\ket{\psi_\theta} = \frac{\ket{0}+e^{i\theta}\ket{1}}{\sqrt{2}}$ with one unknown parameter $\theta$. The optimal choice of measurement basis will be in the $x$-$y$ plane of the Bloch sphere (see \figref{fig:bloch}) because any component in the $z$ axis is not useful in discriminating $\ket{\psi_\theta}$ and $\ket{\psi_{\theta+\mathrm{ d}\theta}}$. Consider a projective measurement onto basis $\{\frac{\ket{0}\pm e^{i\beta}\ket{1}}{\sqrt{2}}\}$. The corresponding probability distribution is $\{\cos^2(\frac{\beta - \theta}{2}),\sin^2(\frac{\beta - \theta}{2})\}$, and the CFI
\begin{align}
    I(M) &= \frac{\big( \partial_\theta(\sin^2(\frac{\beta - \theta}{2})) \big)^2}{\cos^2(\frac{\beta - \theta}{2}) } + \frac{\big( \partial_\theta(\cos^2(\frac{\beta - \theta}{2}))\big)^2}{\sin^2(\frac{\beta - \theta}{2})} \nonumber \\
    &= \sin^2\bigg(\frac{\beta - \theta}{2}\bigg) + \cos^2\bigg(\frac{\beta - \theta}{2}\bigg) = 1. \nonumber
\end{align}
In fact, the CFI is $I(M) = J = 1$ for any $\beta - \theta \notin \frac{\pi}{2}\bZ$, which means (almost) all projective measurements on the $x$-$y$ plane of the Bloch sphere can achieve the QFI. To intuitively understand why all measurements on the $x$-$y$ plane can yield the same performance, we consider the situation when $\theta$ is a small positive angle (see \figref{fig:bloch}) and compare projective measurements onto $\ket{\pm}$ (eigenstates of Pauli-$X$, $\beta = 0$) and $\ket{\pm i}$ (eigenstates of Pauli-$Y$, $\beta = \frac{\pi}{2}$). The probabilities of obtaining $\ket{+}$ and $\ket{-}$ as outcomes are approximately $1 - \theta^2/4$ and $\theta^2/4$, respectively, with corresponding derivatives approximately $\mp\theta/2$. On the other hand, for a measurement in the ${\ket{+i}, \ket{-i}}$ basis, the outcome probabilities are approximately $1/2 \mp \theta/2$, with derivatives around $\mp 1/2$. In the first case, the derivatives of the probabilities are small, but the distribution is highly sensitive to changes in their values. In the second case, the derivatives are large, but the distribution is less sensitive to changes in the values of probabilities. 
These two factors affect the overall estimation precision in opposing ways, ultimately balancing each other out. As a result, the CFI remains the same across all measurements in the $x$–$y$ plane. This shows for pure states the optimal measurement to estimate one parameter is not unique and can spread extensively in the entire space of POVMs, which provides evidence that randomized measurements may be able to achieve a near-optimal performance.

\begin{figure}[tb]
    \centering
    \includegraphics[width=0.75\linewidth]{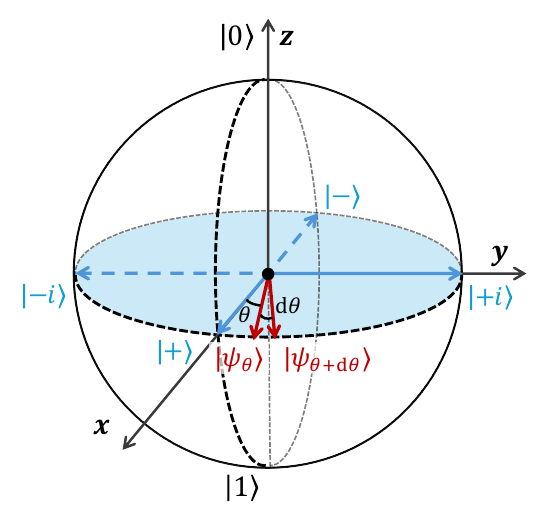}
    \caption{Bloch sphere. Any projective measurement onto the basis in the $x$-$y$ plane is optimal in the task of estimating parameter $\theta$ in $\ket{\psi_\theta} = \frac{\ket{0}+e^{i\theta}\ket{1}}{\sqrt{2}}$, including the $x$-basis $\{\ket{\pm} = \frac{\ket{0}\pm\ket{1}}{\sqrt{2}}\}$ and the $y$-basis $\{\ket{\pm i} = \frac{\ket{0}\pm i\ket{1}}{\sqrt{2}}\}$.}
    \label{fig:bloch}
\end{figure}

To see why the above discussion is still (surprisingly) meaningful in a general $d$-dimensional Hilbert space, we perform the following hand-wavy analysis on the scaling of the CFI with respect to $d$. Consider a $d$-dimensional pure state $\ket{\psi_\theta}$ and a projective measurement $M = \{U\ket{x}\!\bra{x}U^\dagger\}_{x=1}^d$ where $\ket{x}$ is the computational basis in the Hilbert space $\mH$. The CFI
\begin{equation}
    I(M) = \sum_x \frac{\bra{x}U^\dagger (\ket{\partial_\theta\psi_\theta}\!\bra{\psi_\theta} + \ket{\psi_\theta}\!\bra{\partial_\theta\psi_\theta}) U\ket{x}^2}{\braket{x|U^\dagger|\psi_\theta}\! \braket{\psi_\theta|U|x}}.\nonumber
\end{equation}
We can naively compute the expectation values of the denominators and the numerators when $U$ is sampled from the Haar measure. We have 
\begin{gather}
\bE_U[\braket{x|U^\dagger|\psi_\theta} \!\braket{\psi_\theta|U|x}] = \frac{1}{d},\nonumber \\
\bE_U[\bra{x}\!U^\dagger \!(\ket{\partial_\theta\psi_\theta}\!\bra{\psi_\theta}\! +\! \ket{\psi_\theta}\!\bra{\partial_\theta\psi_\theta}) U\!\ket{x}^2] = \frac{J}{2d(d+1)},  \nonumber  
\end{gather}
where we use the property of $2$-design (see \eqref{eq:prop-1}) and the fact that the pure-state QFI is $J = 4(\braket{\partial_\theta\psi|\partial_\theta\psi} - \braket{\psi|\partial_\theta\psi}\braket{\partial_\theta\psi|\psi})$. If naively, in the expression of $I(M)$, we replace all numerators and denominators with their expectation values, we can see the ratio between $I(M)$ and $J$ is a constant independent of $d$. Roughly speaking, as $d$ increases, the derivative of each probability decreases, but the distribution also becomes more sensitive to changes in the values of probabilities---because the number of possible outcomes also increases. For example, compare a Bernoulli distribution $\{1-\theta,\theta\}$ to a $(d+1)$-outcome probability distribution $\{1-\theta,\theta/d,\ldots,\theta/d\}$. Their CFIs with respect to $\theta$ are the same, equal to $1/(\theta(1-\theta))$, and the latter has more numbers of outcomes but smaller derivatives of each probability. In our case, these two factors also balance each other out, keeping the ratio between $I(M)$ and $J$ a constant. As the structure of the randomized measurements does not depend on the specific choice of parameter, the result should extend naturally to multi-parameter estimation. 

The above analysis, though not rigorous, indicates when a measurement is sufficiently random, it may potentially achieve near-optimality for pure state estimation. We prove this rigorously in \thmref{thm:pure} for $3$-design measurements where we explicitly construct a locally unbiased estimator that achieves the QFIM up to a constant. We further generalize this result to \thmref{thm:unitary} and \thmref{thm:low-rank}, where we show near-optimality for {approximately low-rank well-conditioned states}. Here, \emph{{approximately low-rank well-conditioned states}} means states which have a low-rank support on which eigenvalues are above a positive constant and the QFIM is close to the QFIM of the entire state. It is practically relevant, e.g., in situations when we estimate parameters in pure states or low-rank states subject to weak noise. 

We can perform a similar dimension analysis to understand the intuition behind this result. Consider a rank-$r$ state $\rho_\theta = \frac{1}{r}\Pi$ where $\Pi$ is a projector and $\partial_\theta \rho_\theta = - i [H,\rho_\theta]$. Then calculation shows the QFI 
\begin{equation}
    J = \frac{4}{r} \trace(H\Pi H \Pi^\perp),
\end{equation}
where $\Pi^\perp = \id - \Pi$. For $M = \{U\ket{x}\!\bra{x}U^\dagger\}_{x=1}^d$, 
\begin{equation}
    I(M) = \sum_x \frac{\bra{x}U^\dagger (\frac{1}{r} (H\Pi-\Pi H)) U\ket{x}^2}{\braket{x| \frac{1}{r}\Pi |x}}.\nonumber
\end{equation}
For randomized $U$ that is at least $2$-design, the average values of the denominators and the numerators are 
\begin{gather}
\bE_U[\bra{x} U^\dagger \bigg( \frac{1}{r}\Pi \bigg) U\ket{x}] = \frac{1}{d},\nonumber \\
\bE_U[\bra{x}\!U^\dagger \bigg(\! \frac{1}{r} (H \Pi - \Pi H) \!\bigg)  U\!\ket{x}^2] = \frac{J}{2r^2d(d+1)}.  \nonumber  
\end{gather}
When $r$ is a constant, i.e., $\rho_\theta$ is low-rank, we can see the ratio between $I(M)$ and $J$ is also a constant independent of $d$---if we naively replace all numerators and denominators with their expectation values in the expression of $I(M)$. Although the replacement is not mathematically valid, it provides some intuition behind \thmref{thm:unitary} and \thmref{thm:low-rank}. 

One key assumption in the above definition of {approximately low-rank well-conditioned states} is all eigenvalues in its low-rank support must be above a positive constant, the following shows a \emph{counter example} where this assumption is violated and randomized measurements are not near-optimal. Rigorous discussions can be found in \secref{sec:no-go}. Consider a single-qubit state $\rho_f = f \ket{0}\!\bra{0} + (1-f) \ket{1}\!\bra{1}$ where $f$ is to be estimated. The optimal measurement is $M = \{\ket{0}\!\bra{0},\ket{1}\!\bra{1}\}$ that achieves $I(M) = J = 1/(f(1-f))$. When $f$ or $1-f$ is sufficiently small, $J$ can be arbitrarily large. On the other hand, if we consider measurement $M = \{U\ket{0}\!\bra{0}U^\dagger,U\ket{1}\!\bra{1}U^\dagger\}$, we have 
\begin{gather}
    I(M) = \sum_{x=0,1} \frac{\bra{x}U^\dagger (\ket{0}\!\bra{0} - \ket{1}\!\bra{1}) U\ket{x}^2}{\braket{x|(f \ket{0}\!\bra{0} + (1-f) \ket{1}\!\bra{1})|x}},\nonumber
    \end{gather}
and for randomized $U$ in the $2$-dimensional Hilbert space (that is at least $2$-design), 
\begin{gather}
\bE_U[\braket{x|(f \ket{0}\!\bra{0} + (1-f) \ket{1}\!\bra{1})|x}] = \frac{1}{2},\nonumber \\
\bE_U[\bra{x}U^\dagger (\ket{0}\!\bra{0} - \ket{1}\!\bra{1}) U\ket{x}^2] = \frac{2}{3},  \nonumber  
\end{gather}
both of which are independent of $f$. It suggests that as $f \rightarrow 0$, the CFI of randomized measurements will not approach infinity---thus an infinitely large gap exists between the CFI and the QFI. In particular, the counter example is a qubit estimation example, which means even for states constrained in a finite-dimensional space, randomized measurements cannot be always near-optimal.

For general mixed states, not only randomized measurements cannot be near-optimal, near-optimal measurements (that only acts on individual copies of quantum states) might not exist at all. A typical example is the full tomography of mixed states, where the measurement incompatibility issue cannot be avoided unless measurements on many copies of states are allowed---which is beyond the discussion of our work. Thus, we introduce the concept of \emph{weak near-optimality} to describe measurements that perform almost the best among all other individual measurements using a natural figure of merit as detailed in \secref{sec:near-opt}. We found in \thmref{thm:full} that $3$-design measurements are weakly near-optimal when estimating all parameters describing a rank-$r$ state that is well-conditioned on its support---we call this the estimation of \emph{full-parameter} rank-$r$ well-conditioned states. Our result is consistent with other existing results in the field of state tomography---e.g., Clifford measurements are optimal for learning all Pauli observables of general mixed states~\cite{huang2020predicting,chen2022exponential}. 

When $m$, the number of parameters to be estimated in general mixed states, is not maximal, randomized measurements are usually not weakly near-optimal. One typical situation is when the QFIM is attainable because the optimal POVMs with respect to different parameters commute. For example, we can consider an $n$-qubit state 
\begin{equation}
    \rho_\theta = \frac{\id + \sum_{i\in \mS^\circ} \theta_i P_i}{2^n}, 
\end{equation}
where $\mS^\circ$ is a stabilizer group~\cite{nielsen2002quantum} excluding $\id$ and we consider parameter estimation at $\theta_i = 0$. The number of parameters $m = 2^n - 1$ (which is far from the maximal number of parameters $4^n-1$), 
the QFIM $J_{ij} = \delta_{ij}$ is achievable by projective measurement onto the common eigenstates of the stabilizer group. On the other hand, consider a randomized measurement $M = \{p_U U\ket{x}\!\bra{x}U^\dagger\}_{x,U}$ where $\ket{x}$ is the computational basis and $\{p_U,U\}$ is a unitary $2$-design, we have
\begin{equation}
    I(M)_{ij} = \sum_{U,x} p_U \braket{x|U^\dagger P_i U|x}\!\braket{x|U^\dagger P_j U|x} = \frac{\delta_{ij}}{2^n+1},
\end{equation}
which is a factor $1/(2^n+1)$ away from optimal. In general, it is open for which types of mixed states randomized measurements are weakly near-optimal.

Below, we will review concepts in multi-parameter quantum metrology and some previous results in \secref{sec:prelim-0}. We provide the exact statements and proofs of our theorems in \secref{sec:pure}--\secref{sec:full}. In \secref{sec:example}, we consider specific examples where we show how to apply our result in estimating fidelity of pure states and in Hamiltonian estimation, and also prove rigorously randomized measurements are not near-optimal when estimating fidelity of some types of mixed states.  



\section{Preliminaries}
\label{sec:prelim-0}

\subsection{Multi-parameter estimation}
\label{sec:prelim}


Here we provide a detailed overview of multi-parameter estimation. 
First, an estimator $(\htheta,M)$ is called \emph{locally unbiased} at $\theta = \theta^0$ if and only if 
\begin{gather}
\label{eq:l-u-condition-1}
    \theta^0 = \sum_{x} \htheta(x) \trace(\rho_\theta M_x) \Big|_{\theta = \theta^0},\\
\label{eq:l-u-condition-2}
    \delta_{ij} = \frac{\partial}{\partial \theta_i}\sum_{x} \htheta_j(x) \trace(\rho_\theta M_x) \Big|_{\theta = \theta^0},\;\forall i,j. 
\end{gather}  
Throughout this paper, when there is no ambiguity, we will usually not distinguish between $\theta^0$ and $\theta$, and assume all functions are evaluated at $\theta$ for notation simplicity. We will only explicitly use the notation $\theta^0$ when we need to distinguish between the local point $\theta^0$ and the true unknown parameter, which will be denoted by $\theta$. 

In classical statistics, the MSEM of any locally unbiased estimator at $\theta$ is bounded below by the inverse of the CFIM (see \eqref{eq:CRB}). $I(M)$ is a positive semidefinite matrix in $\bR^{m\times m}$ and we further assume in this work that $I(M)$ is strictly positive, which means each parameter is identifiable (unless stated otherwise). 
This forbids redundant parameterization and puts an upper bound on the number of allowed parameters.
Note that $I(M)$ depends on the state $\rho_\theta$ and the value of $\theta$ but they are omitted for simplicity, as is the case with many other concepts defined below. 
The CRB at $\theta^0$ is saturable via the following locally unbiased estimator~\cite{rao1973linear,kay1993fundamentals,lehmann2006theory,cox2017inference,casella2002statistical} defined by 
\begin{equation}
\label{eq:opt}
    \htheta^{\opt}_i(y;\theta^0) = \theta_i^0 + \sum_x (\alpha_{i,x}|_{\theta = \theta^0}) \hat{p}_x(y),
\end{equation}
where $(~; \star)$ indicates dependence on $\star$, $\alpha_{i,x} = \sum_{i} (I(M)^{-1})_{ij} \frac{\partial_j p_x}{p_x} $ and $\partial_i$ is used as a shorthand for $\frac{\partial}{\partial \theta_i}$. $\hat{p}_x(y) = \delta_{xy}$ where $y$ is the measurement outcome. $\bE[\hat{p}_x] = p_x = \trace(\rho_\theta M_x)$ and $\bE[\hat{p}_x \hat{p}_y] =\delta_{xy} p_x$. It achieves the CRB, i.e., 
\begin{equation}
    V(M,\hat\theta^\opt) = I(M)^{-1}. 
\end{equation}
Given $N$ copies of $\rho_\theta$, the CRB is asymptotically saturable for large $N$ using the (asymptotically unbiased) maximum likelihood estimator $\htheta^{(N)}_{\textsc{mle}}$~\cite{rao1973linear,kay1993fundamentals,lehmann2006theory,cox2017inference,casella2002statistical}---$\sqrt{N}(\hat{\theta}^{(N)}_{\text{MLE}} - \theta)$ converges in distribution to a normal distribution centered around $0$ with variance equal to $I(M)^{-1}$ as $N \rightarrow \infty$. Note that $I(M)$ here is additive---$I(M^{\otimes N}) = N I(M)$, where we use $I(M^{\otimes N})$ to denote the CFIM when measuring $\rho_\theta^{\otimes N}$ using POVM $M^{\otimes N} = \{\bigotimes_{i=1}^N M_{x_i}\}_{(x_1,x_2,\ldots,x_N)}$. The factor of $N$ naturally occurs from the classical central limit theorem, where the estimation variance is improved by a factor of $N$ when $N$ i.i.d. samples are taken. 

The CRB is a purely classical result. It provides a lower bound on $V(M,\htheta)$ that depends on the choice of POVM. The QCRB~\cite{helstrom1967minimum,helstrom1968minimum,helstrom1976quantum,holevo2011probabilistic,braunstein1994statistical,barndorff2000fisher,paris2009quantum} (see \eqref{eq:QCRB}) provides a more general lower bound when the POVM can be chosen freely. 
$J$ is the QFIM as a function of $\rho_\theta$ and is additive with respect to $\rho_\theta$. It is defined by 
\begin{equation}
    J_{ij} = \frac{1}{2}\trace(\rho_\theta \{L_i,L_j\}),
\end{equation}
where $\{\cdot,\cdot\}$ is the anti-commutator and $L_i$ are Hermitian operators called symmetric logarithmic derivative (SLD) operators, defined through 
\begin{equation}
    \frac{\partial \rho_\theta}{\partial \theta_i} = \frac{1}{2}(L_i \rho_\theta + \rho_\theta L_i). 
\end{equation}
$L_i$ is not uniquely defined by the equation above, but $J$ is invariant under different choices of $L_i$. In this work, we fix the definition of $L_i$ as below for future use. 
\begin{equation}
    L_i := 2 \sum_{jk,\lambda_j+\lambda_k>0} \frac{\braket{\psi_j|\partial_i\rho_\theta|\psi_k}}{\lambda_j+\lambda_k} \ket{\psi_j}\bra{\psi_k},
\end{equation}
where $\sum_j \lambda_j \ket{\psi_j}\bra{\psi_j} = \rho_\theta$ is the spectral decomposition of $\rho_\theta$. The QFIM $J$ can be calculated as 
\begin{equation}
    J_{ij} = 2 \sum_{k\ell,\lambda_k+\lambda_\ell>0} \frac{\Re[\braket{\psi_k|\partial_i\rho_\theta|\psi_\ell}\braket{\psi_\ell|\partial_j\rho_\theta|\psi_k}]}{\lambda_k+\lambda_\ell}.
\end{equation}

For single-parameter estimation, i.e., when $m = 1$, $J$ and $I(M)$ are scalars, and $J$ is attainable~\cite{braunstein1994statistical}:
\begin{equation}
\label{eq:single}
    J = \max_{M:\mathrm{POVM}}I(M).
\end{equation}
Therefore, the QCRB (\eqref{eq:QCRB}) is asymptotically attainable, making the QFI $J$ a perfect figure of merit to characterize the performance of local parameter estimation on state $\rho_\theta$. For multi-parameter estimation, however, $J$ may not be attainable by any POVM due to the measurement incompatibility issue.

\subsection{Previous Results}

\subsubsection{Holevo Cram\'{e}r--Rao bound}

To address the measurement incompatibility issue, it helps to consider minimization of the \emph{weighted mean squared error} (WMSE), $\trace(W V(M,\htheta))$, where $W$ is a cost matrix that is positive semidefinite in $\bR^{m \times m}$. The QCRB automatically provides a lower bound on the WMSE, 
\begin{equation}
    \trace(W V(M,\htheta)) \geq \trace(W J^{-1}) =: C_{\rm F}(W). 
\end{equation}
The HCRB~\cite{holevo2011probabilistic,kahn2009local,yamagata2013quantum,yang2019attaining} provides a tighter lower bound which states 
\begin{multline}
    \trace(W V(M,\htheta)) \geq C_{\rm H}(W):= \\ \! \min_{\vX:=\{X_i\}_i} \!\trace(W \Re[Z(\vX)]) + \norm{\sqrt{W} \! \Im[Z(\vX)] \! \sqrt{W}}_1, 
\end{multline}
where $X_i$ are Hermitian operators acting on $\mH$ satisfying $\trace\big(X_i \partial_j \rho_\theta \big) = \delta_{ij}$ and $Z(\vX)_{ij} = \trace(\rho_\theta X_i X_j)$. 
Both $C_{\rm F}(W)$ and $C_{\rm H}(W)$ are additive with respect to $\rho_\theta$.  The HCRB is asymptotically approachable through a sequence of \emph{collective measurements} $M^{(N)}$ acting on $\rho_\theta^{\otimes N}$ when the state and estimators satisfy certain regularity conditions~\cite{kahn2009local,yamagata2013quantum,yang2019attaining}. The HCRB is computable through semidefinite programming~\cite{albarelli2019evaluating} and can be approximated by the QCRB up to a factor of two as~\cite{tsang2020quantum,carollo2019quantumness}
\begin{equation}
\label{eq:qcrb-hcrb}
    C_{\rm F}(W) \leq C_{\rm H}(W) \leq 2 C_{\rm F}(W). 
\end{equation}

The HCRB is achievable only through collective measurements in general. When only \emph{individual measurements} on single copies are available, we denote the minimum WMSE over locally unbiased measurements by 
\begin{align}
    C(W) &:= \min_{(M,\htheta):\text{ locally unbiased}} \trace(W V(M,\htheta))\\
    &= \min_{M}\, \trace(W I(M)^{-1}). 
\end{align}
Note that $C(W)$ cannot be further improved even when separable, adaptive measurements across multiple copies of $\rho_\theta$ are available~\cite{nagaoka2005parameter,gill2000state}. 
While $C(W)$ is strictly larger than $C_{\rm H}(W)$ in general, for pure states, $C(W) = C_{\rm H}(W)$, and the HCRB is achievable through individual measurements~\cite{matsumoto2002new}.

\subsubsection{Gill--Massar bound}
\label{sec:GM}

Collective measurements (as experimentally demonstrated in~\cite{hou2018deterministic}) can in general be much more powerful than individual measurements. It can be seen from the Gill--Massar (GM) bound~\cite{gill2000state} (see also \appref{app:GM}) that applies to all individual measurements. It states that for any $d$-dimensional quantum state $\rho_\theta$ and any (individual) POVM $\tM$,
\begin{equation}
\label{eq:GM}
    \trace(J^{-1}I(\tM)) \leq d-1. 
\end{equation}
It sets a lower bound on $\trace(J I(\tM)^{-1})$ through Cauchy--Schwarz inequality, 
\begin{equation}
\label{eq:lower}
    \trace(J I(\tM)^{-1}) \geq \frac{\trace(\id)^2}{\trace(J^{-1}I(\tM))} \geq \frac{m^2}{d-1}. 
\end{equation}

For example, if we take $W = J$, $C_{\rm H}(J) \leq 2m$ (easily seen from \eqref{eq:qcrb-hcrb}), while $C(J) \geq m^2/(d-1)$~\cite{gill2000state}. When the number of unknown parameters is $\Theta(d^2)$, e.g., in tomography, we have $C_{\rm H}(J) = O(d^2)$ while $C(J) = \Omega(d^3)$ which is significantly larger than $C_{\rm H}(J)$.

\subsubsection{Fisher-symmetric measurements for pure states}

Given a parameterized pure state $\rho_\theta$, it was known that there exists a rank-one measurement $M$ (with $2d-1$ measurement outcomes) such that the Fisher-symmetry condition is satisfied~\cite{li2016fisher}, i.e., 
\begin{equation}
\label{eq:f-s-pure}
    I(M) = \frac{1}{2} J,
\end{equation}
where $1/2$ is the best possible constant for state tomography, i.e., when the number of parameters $m = 2(d-1)$. One can see this from the GM bound---if $I(M) = c J$ from some constant $c$, then $
\trace(J^{-1} I(M)) = c m = 2c(d-1) \leq d-1$. 
(Note that although \eqref{eq:f-s-pure} was proven only for state tomography, it automatically holds for arbitrary $m$.) 

The above measurement~\cite{li2016fisher} is a function of the state $\rho_\theta$ and in practical implementation, must be adjusted adaptively based on the accumulated knowledge of the states~\cite{vargas2024near}. In addition, both the computational complexity and implementation complexity of this approach can be costly. 
To overcome the limitation on state-dependence, researchers studied universally Fisher-symmetric measurements that are independent of states. It was known the infinite-outcome Haar random measurements~\cite{hayashi1998asymptotic,Zhu_HJthesis,zhu2014quantum,lu2025quantum} are universally Fisher-symmetric and can achieve \eqref{eq:f-s-pure} for all pure states universally. However, they are experimentally infeasible because they require infinite amount of classical randomness and exponential gate complexity. 
In fact, \cite{zhu2018universally} showed any measurement $M$ satisfying \eqref{eq:f-s-pure} for all pure states must contains infinite number of measurement outcomes (and thus are not feasible). They proposed the $2$-design collective measurements on two identical copies of the state as a workaround (as also experimentally demonstrated for the single-qubit case~\cite{hou2018deterministic}). They showed that 
\begin{equation}
    I[\rho_\theta^{\otimes 2},M] = J[\rho_\theta^{\otimes 2}],
\end{equation}
when $M$ is a $2$-design measurement on $2$ copies of any pure state $\rho_\theta$. Here, we use $I[\star,M]$ ($J[\star]$) to denote the CFIM (QFIM) of the state $\star$, which is usually omitted in the expressions. In practice, implementing a two-copy measurement requires simultaneous access to and precise cross-device control of two identical quantum sensors. Consequently, a more experimentally accessible single-copy (individual) measurement protocol remains highly desirable. Our result partially resolves this challenge, providing feasible finite-outcome individual measurement protocols by tolerating a constant-factor suboptimality.


\subsection{Definition of Near-Optimality}
\label{sec:near-opt}

In this work, we will analyze the performance of randomized individual measurements for parameter estimation at a local point of $\theta$. We define optimality conditions of measurements for families of quantum states where a constant-factor suboptimality is tolerated, in which case we use the term ``near-optimal'' to represent a potential constant-factor difference from the exact optimality. The near-optimality property is particularly important in situations where the optimal estimation precision can scale up or down with respect to system variables, e.g., the system dimension, the noise strength, the number of unknown parameters, etc. and near-optimality guarantees these scalings to be preserved across the entire range of system variables. 

We first define quantum measurements that nearly saturate the QCRB by the following. 
\begin{definition}[Near-optimality]
A family of POVMs is near-optimal for a family of parameterized quantum states, if there exists a positive constant $c$ (independent of the state or POVM) such that for any state in the family and its corresponding POVM $M$, 
\begin{equation}
\label{eq:def-near-optimal}
    I(M) \succeq c J.
\end{equation} 
\end{definition}

Given a family of general mixed states, near-optimal individual measurements may not exist. For example, for maximally mixed states $\rho = \frac{\id}{d}$ with maximal number of unknown parameters, $C_{\rm F}(J) = O(d^2)$ while $C(J) = \Omega(d^3)$, implying the $\Omega(d)$ gap between $I(M)$ and $J$. Therefore, we need a less stringent definition. 
\begin{definition}[Weak near-optimality]
\label{def:weak}
A family of POVMs is weakly near-optimal for a family of parameterized quantum states, if there exists positive constants $c_1,c_2$ (independent of the state or POVM) such that for any state in the family and its corresponding POVM $M$, 
\begin{gather}
\label{eq:def-weak-near-optimal-1}
    \trace(J \cdot I(M)^{-1}) \leq c_1 \inf_{\tM:\mathrm{POVM}}\trace(J \cdot I(\tM)^{-1}),\\
\label{eq:def-weak-near-optimal-2}
    I(M) \succeq c_2 \frac{\trace(J^{-1} I(M))}{m} J. 
\end{gather} 
\end{definition}
\noindent Recall $m$ is the number of parameters.
\eqref{eq:def-weak-near-optimal-1} implies the weak near-optimality measurement minimizes the WMSE up to a constant when the cost matrix is taken as the QFIM $J$. It guarantees the average near-optimal performance of $M$. \eqref{eq:def-weak-near-optimal-2} further guarantees the measurement provides an overall near-optimal estimation for all parameters, excluding the situation where the precision of some parameter is extremely low. We further explain and justify this definition in \appref{app:weak-optimal}. We prove whenever a near-optimal measurement exists, weak near-optimality is equivalent to near-optimality. We also explain how it functions as a proper generalization of the Fisher-symmetry condition~\cite{li2016fisher,zhu2018universally,vargas2024near}. 

Below we will see randomized measurements are (weakly) near-optimal in multiple interesting scenarios. 

\section{Pure states}
\label{sec:pure}

In this section, we consider randomized measurements for pure state metrology and obtain the following result. 
\begin{theorem}
\label{thm:pure}
When $M$ is a POVM given by a $3$-design, for any pure state $\rho_\theta$ and any number of parameters $m$,  
\begin{equation}
\label{eq:pure-bound}
    I(M) \succeq \frac{d+2}{4(d+1)} J \succeq \frac{1}{4} J. 
\end{equation}
\end{theorem}

Here we use randomized measurements whose measurement basis forms a $3$-design. Specifically, we say a randomized measurement is given by a $3$-design if $M = \{ M_s = q_s d  \ket{s}\bra{s} \}_s$, where $q_s > 0$ and $\sum_s q_s \ket{s}\bra{s} = \id/d$. $\{q_s,\ket{s}\}$ is called a (complex projective) $3$-design~\cite{renes2004symmetric,ambainis2007quantum} if and only if 
\begin{equation}
\label{eq:$3$-design}
    \sum_s q_s (\ket{s}\bra{s})^{\otimes 3} = \int_{\psi} (\ket{\psi}\bra{\psi})^{\otimes 3} \mathrm{d}\psi,
\end{equation} 
where $\mathrm{d}\psi$ is the (complex projective) Haar measure. One natural way to implement a $3$-design measurement is to implement a random unitary $U$ from a unitary $3$-design~\cite{dankert2009exact,roy2009unitary,gross2007evenly} and then implement projection onto the computational basis. A unitary ensemble $\{q_i, U_i\}$ is a unitary $3$-design if $\sum_i q_i U_i^{\otimes 3} \otimes U_i^{\dagger \otimes 3} = \int \mathrm{d}U U^{\otimes 3} \otimes U^{\dagger \otimes 3}$ where $\mathrm{d}U$ is the Haar measure. For multi-qubit systems, random Clifford circuits form a unitary $3$-design~\cite{webb2015clifford,zhu2017multiqubit,kueng2015qubit}. 

\thmref{thm:pure} implies for pure states, the QFIM is always saturable by $3$-design randomized measurements up to a constant factor, i.e., they are near-optimal for pure states. 
This implies the \emph{universal} optimality of $3$-design randomized measurements for estimating any number of parameters in any pure states. Even in the context of single-parameter estimation, our method is still useful in providing an explicit measurement protocol that is always near-optimal for any parameterized pure states. In particular, the choice of randomized measurements do not rely on the knowledge of states or cost matrices (like in~\cite{matsumoto2002new}), thus fit into the ``measure first, ask questions later'' framework~\cite{elben2023randomized} and substantially lower the complexity of experiments.




Before we prove \thmref{thm:pure}. We first provide two useful properties~\cite{dankert2009exact,roberts2017chaos} of $3$-designs that we will use later, 
\begin{gather}
\sum_s q_s \ket{s}\bra{s} \braket{s|A|s} = \frac{\trace(A)\id + A}{(d+1)d}, \label{eq:prop-1}\\
\begin{split}
    &\sum_s q_s \ket{s}\bra{s} \braket{s|B|s}\braket{s|C|s} = \\
&\;\frac{1}{(d+2)(d+1)d} \bigg( \big(\trace(BC)+\trace(B)\trace(C)\big)\id + \\
&\qquad \qquad \qquad \qquad \;\trace(B)C+\trace(C)B+ \{B,C\} \bigg).
\end{split}
\label{eq:prop-2}
\end{gather}

\begin{proof}[Proof~of~{\thmref{thm:pure}}]
    We prove \thmref{thm:pure} by explicitly constructing a set of locally unbiased estimators $\{\htheta_i\}_i$ that satisfy 
    \begin{equation}
    \label{eq:V-J}
        V = \frac{4(d+1)}{d+2} J^{-1},
    \end{equation}
    where $V:= V(M,\htheta)$. 
    This then implies \eqref{eq:pure-bound} from the CRB (\eqref{eq:CRB}) and the fact that $A \succeq B$ is equivalent to $A^{-1} \preceq B^{-1}$ for any positive matrices $A,B$. 
    To construct locally unbiased estimators, we first introduce an unbiased estimator of $\rho_\theta$ using classical shadows~\cite{huang2020predicting}, 
    \begin{equation}
    \label{eq:rho-estimator}
        \hrho(s) = \mM^{-1}(\ket{s}\bra{s}) = (d+1)\ket{s}\bra{s} - \id, 
    \end{equation}
    where $\mM(\cdot) = \frac{(\cdot) + \trace(\cdot) \id }{d+1}$, and $\mM$ is the inverse map of $\mM$ satisfying $\mM^{-1}(\cdot) = (d+1)(\cdot) - \id$. $\hrho$ is an unbiased estimator of $\rho_\theta$ because 
    \begin{gather}
    \bE[\ket{s}\bra{s}] := \sum_s \trace(\rho_\theta M_s) \ket{s}\bra{s} = \mM(\rho_\theta),
    \end{gather}
    from \eqref{eq:prop-1}, and then 
    \begin{align}
    \bE[\hrho(s)] &=  \bE[\mM^{-1}(\ket{s}\bra{s})] \nonumber\\
    &= \mM^{-1}(\bE[(\ket{s}\bra{s})]) = \rho_\theta. 
    \end{align}
    
    We then define the locally unbiased estimators at any local point $\theta = \theta^0$ by 
    \begin{equation}
    \label{eq:theta-estimator}
        \htheta_i(s;\theta^0) = \theta_i^0 + \trace\big( (X_i |_{\theta = \theta^0}) \hrho(s)\big),
    \end{equation}
    where $X_i$ are Hermitian operators which we will call \emph{deviation observables} that satisfies the locally unbiasedness conditions 
    \begin{equation}
    \label{eq:l-u-condition-X}
        \trace(\rho_\theta X_i)|_{\theta = \theta^0} = 0 ,\;\trace((\partial_i\rho_\theta) X_j)|_{\theta = \theta^0} = \delta_{ij} ,\;\forall i,j. 
    \end{equation}
    Here $\{X_i\}_{i=1}^m$ are named deviation observables because they are observables used to estimate the deviation of $\theta$ from $\theta^0$. $X_i$ can vary for different $\theta$ but the dependence is omitted for simplicity. We will refer to estimators of form \eqref{eq:theta-estimator} as \emph{local shadow estimators}, a type of locally unbiased estimator defined through deviation observables $\{X_i\}_i$.

    To make sure \eqref{eq:l-u-condition-X} is satisfied, we choose 
    \begin{equation}
    \label{eq:def-X}
        X_i = \sum_j (J^{-1})_{ij} L_j,
    \end{equation}
    where $L_j$ are SLDs. They satisfy $\trace(\rho_\theta L_j) = 0$ and $J_{ij} = \frac{1}{2}\trace(\rho_\theta\{L_i,L_j\})$ by definition. We can verify \eqref{eq:l-u-condition-1} and \eqref{eq:l-u-condition-2} at $\theta = \theta^0$ through
    \begin{align}
        &\bE[\htheta_i(s;\theta^0)]|_{\theta = \theta^0} = \theta_i^0 + \trace(X_i \rho_\theta)|_{\theta = \theta^0} = \theta^0_i, \\
        &(\partial_j \bE[\htheta_i(s;\theta^0)])|_{\theta = \theta^0}
        = \trace(X_i|_{\theta = \theta^0} (\partial_j \rho_\theta)|_{\theta = \theta^0}) \nonumber\\
        &\qquad = \frac{1}{2}\sum_k (J^{-1})_{ik} \trace(\rho_\theta\{L_k,L_j\})\big|_{\theta = \theta^0} = \delta_{ij}. 
    \end{align}
    The MSEM obtained from this choice of deviation observables is at most a constant factor away from the inverse of the QFIM, as we desire. For any $X_i$ and corresponding estimators (\eqref{eq:theta-estimator}), we have 
    \begin{align}
        V_{ij} &= \sum_s \trace(\rho_\theta M_s) (\htheta_i(s;\theta) - \theta_i)(\htheta_j(s;\theta) - \theta_j) \nonumber\\
        &= \sum_s q_s d \braket{s|\rho_\theta|s}\braket{s|\mM^{-1}(X_i)|s}\braket{s|\mM^{-1}(X_j)|s}\nonumber \\
&= \frac{d+1}{d+2}\Big(\trace(X_iX_j) + \trace(\rho_\theta\{X_i,X_j\})\Big)  \nonumber\\
&\qquad \qquad \qquad \qquad \qquad  - \frac{1}{d+2} \trace(X_i)\trace(X_j), \label{eq:formula} 
\end{align}
where in the second equality we use the dual map of $\mM^{-1}$ is itself. The third equality is obtained by plugging in $\mM^{-1}(X_i) = (d+1)X_i - \trace(X_i) \id$ and applying \eqref{eq:prop-2} and the first equality in \eqref{eq:l-u-condition-X}. 

\eqref{eq:formula} holds for any choice of $X_i$ and any mixed $\rho_\theta$. In the situation of pure states $\rho_\theta = \ket{\psi}\bra{\psi}$ and definitions \eqref{eq:def-X} of $X_i$, \eqref{eq:formula} can be further simplify by noticing that
\begin{gather}
    L_i = 2(\ket{\partial_i \psi}\bra{\psi} + \ket{\psi}\bra{\partial_i \psi}), \\
    \trace(L_i) = \trace(\rho_\theta L_i) = 0,\\
    J_{ij} = 4 \Re[\braket{\partial_i\psi|\partial_j\psi} - \braket{\psi|\partial_i\psi}\braket{\partial_j\psi|\psi}],\\ 
    \trace(L_iL_j) = 2 J_{ij}. 
\end{gather}
Applying \eqref{eq:formula} and noticing the second term is zero, we have 
\begin{align}
    (JVJ)_{ij} 
    &= \frac{d+1}{d+2} \left( \trace(L_iL_j) + \trace(\rho_\theta\{L_i,L_j\}) \right)\nonumber \\
    &= \frac{4(d+1)}{d+2} J_{ij},
\end{align}
which implies \eqref{eq:V-J} if we multiply both sides of the equation by $J^{-1}$ on the left and right. 
\end{proof}

Traditionally, given any POVM $\{M_x\}_x$ (including the randomized measurement), it is always optimal to choose the locally unbiased estimator in \eqref{eq:opt} because it minimizes $V$ and achieves $V = I(M)^{-1}$. 
However, a direct calculation of $I(M)$ is challenging because the randomization appears in both the denominator and the numerator of the expression. In contrast, our proof, which uses the local shadow estimator, only involves up to the third moment of the randomization (when calculating \eqref{eq:formula}), where $3$-designs are sufficient. It is open whether $2$-design measurements will be sufficient, and whether \thmref{thm:pure} can be further improved using other types of locally unbiased estimators and higher-order designs.

Note that in the proof of \thmref{thm:pure} (and in the proofs of other theorems below), we introduce the local shadow estimators (\eqref{eq:theta-estimator}), which are locally unbiased estimator that achieves the near-optimal estimation precision. For practical applications, one needs to first compute the deviation observables $X_i$ through \eqref{eq:def-X} and then performs randomized measurements on $\rho_\theta$ to estimate $\trace\big( (X_i |_{\theta = \theta^0}) \hrho(s)\big)$ using $\hat\rho(s)$, the snapshots of classical shadow. Two caveats remain: (1)~the deviation observable is a function of $\theta^0$ which is not known a priori; (2)~$\trace(X_i \hrho(s))$ may not be efficiently computable, both due to the difficulty in computing $X_i$ and the difficulty in obtaining $\braket{s|X_i|s}$. 

To tackle caveat~(1), the two-step approach~\cite{barndorff2000fisher,gill2000state,yang2019attaining} is often adopted in quantum metrology, where a coarse estimate of $\theta$ is obtained first and $X_i$ is then evaluated at the coarse estimate in order to calculate $\trace(X_i \hrho(s))$. \secref{sec:example} contains an example of fidelity estimation where we numerically demonstrate the effectiveness of the two-step approach. 

To address caveat~(2), we need to leverage special structures of $X_i$ and measurements. Below we consider a $n$-qubit system where the $3$-design measurement is implemented through random Clifford sampling. Assume $\ket{\psi_\theta}$ and its derivatives $\ket{\partial_i \psi_\theta}$ are all restricted to a subspace spanned by a polynomial number of computational basis, which is a common scenario in quantum metrology. Then the total number of unknown parameters is at most a polynomial in $n$, which means $X_i$ is efficiently computable because both the SLD $L_i$ and the inverse of the QFIM $J$ are efficiently computable. Furthermore, $X_i$ can be decomposed into a polynomial number of computational basis operators, i.e., $X_i = \sum_{x,y \in \{0,1\}^n} f_{x,y} \ket{x}\!\bra{y}$ where $f_{x,y} \neq 0$ only for polynomial numbers of $x$ and $y$. Then $\trace(X_i \hrho(s))$ is efficiently computable because $\ket{s}$ is a stabilizer state and $\trace(\ket{x}\!\bra{y} \hrho(s))$ is efficiently computable. 
In general, to efficiently compute $X_i$, we need both the SLD $L_i$ and the inverse of the QFIM $J$ to be efficiently computable. The first condition is usually easily satisfiable for pure states that have efficient classical representations. The second condition is tractable in situations where $J$ is diagonal or block-diagonal (see examples in \secref{sec:fid_pure} and \secref{sec:Hamt}). To efficiently compute $\trace(X_i \hrho(s))$, further structures on $X_i$ are needed, e.g., 
a stabilizer formalism that efficiently represents $X_i$ (see an example in \secref{sec:Hamt}), etc. We leave the study of general algorithms for computing $X_i$ and sampling $\braket{s|X_i|s}$ for future work.

\section{Approximately Low-rank states}
\label{sec:low-rank}

In this section, we consider a generalization of the previous result for pure states to \emph{{approximately low-rank well-conditioned states}}. They include not only low-rank states, but also situations where the QFIM of some complex quantum states can be well approximated by confinement in a low-rank subspace. Previously, the optimality of low-rank states and their approximate (noisy) versions was not investigated in the context of either minimizing the WMSE or the local state tomography, to the best of our knowledge.

We define {approximately low-rank well-conditioned states} by the following. 
\begin{definition}[{Approximately low-rank well-conditioned state}]
\label{def:low-rank}
    A parameterized state $\rho_\theta$ is called an $(\mu,c)$-approximately low-rank well-conditioned state if it satisfies the following two conditions:
    \begin{enumerate}[wide, labelwidth=!,itemindent=!,labelindent=0pt, leftmargin=0em, label={(\arabic*)}, parsep=0pt]
        \item There is a constant $\mu$ such that there is at least one eigenvalue of $\rho_\theta$ that is at least $\mu$. 
        \item Let $\Pi$ be the projector onto the subspace spanned by eigenstates of $\rho_\theta$ whose corresponding eigenvalues are at least $\mu$. $\Pi_\perp = \id - \Pi$ and $p = \trace(\rho_\theta \Pi_\perp)$. There exists a constant $c$ such that 
        \begin{equation}
            J - p J^\perp - J^p \succeq c J.
        \end{equation}
        $J^\perp$ is the QFIM after post-selecting on subspace $\Pi_\perp$, i.e., the QFIM of $\frac{\Pi_\perp \rho_\theta \Pi_\perp}{\trace(\rho_\theta \Pi_\perp)}$ and $J^p$ is the CFIM of a Bernoulli distribution $\{p,1-p\}$, i.e., $J^p_{ij} = \frac{(\partial_i p)(\partial_j p)}{p(1-p)}$. 
    \end{enumerate}
\end{definition}
Condition (1) above guarantees the existence of a low-rank approximation of the original state after projection to a low-rank subspace $\Pi$. Condition (2) guarantees the information of the unknown parameters are largely preserved after excluding the information from the subspace $\Pi_\perp$. This condition is necessary because we do not want the QFIM to concentrate on some high-dimensional subspace $\Pi_\perp$ so that access to states in $\Pi_\perp$ is required in which case any low-rank approximation will fail. Finally, it is worth noting that our definition implicitly assumes the states are well-conditioned within the subspace $\Pi$, i.e., the ratio between the largest and the smallest eigenvalues inside $\Pi$ are bounded (because their eigenvalues are above a constant $\mu$). This is necessary because otherwise randomized measurements may not be near-optimal (see \secref{sec:no-go}). Pure states are a special case of {approximately low-rank well-conditioned states} with $\mu = 1$ and $c = 1$. 

Below we first provide a proof of the near-optimality of randomized measurements when $\theta$ are encoded in a unitary rotation (which, in local parameter estimation, equivalently means the derivatives of all eigenvalues are zero). 
Then we provide generalization to all {approximately low-rank well-conditioned states}. The case of unitary rotations has a better constant and also serves as a perfect warm-up before introducing the general case. 

\begin{theorem}
\label{thm:unitary}
    Consider a $(\mu,c)$-approximately low-rank well-conditioned state $\rho_\theta$ satisfying $\partial_j \lambda_j = 0$ for all of its eigenvalues $\lambda_j$.
    Let $M$ be a POVM given by a $3$-design. We have 
    \begin{equation}
        I(M) \succeq \frac{d+2}{d+1}\cdot \frac{\mu c}{2\mu + 2}  J. 
    \end{equation}
\end{theorem}

\begin{proof}
    We prove the theorem by finding locally unbiased estimators such that 
    \begin{equation}
    \label{eq:V-J-3}
        V = V(M,\htheta) \preceq \frac{d+2}{d+1}\cdot \frac{2\mu + 2}{\mu c}  J^{-1}.
    \end{equation}
    We adopt again local shadow estimators defined in \eqref{eq:theta-estimator} but will choose a different set of deviation observables $X_i$. They will satisfy $\Pi_\perp X_i \Pi_\perp = 0$ which plays a crucial role in bounding $V$, as we will see later. In the following, all functions are evaluated at some local point of $\theta$ implicitly. We define the unnormalized state $\rho_\perp := \Pi_\perp \rho_\theta \Pi_\perp$ for notation simplicity. We will also omit $\theta$ in $\rho_\theta$. 
    
    Let $L_i$ be the SLDs of $\rho$. We define 
    \begin{align}
        \tL_i &= L_i - \Pi_\perp L_i \Pi_\perp \\
        &= 2 \sum_{jk:\lambda_{j}\geq \mu \text{~or~} \lambda_{k}\geq \mu} \frac{\braket{\psi_j|\partial_i\rho|\psi_k}}{\lambda_j+\lambda_k} \ket{\psi_j}\bra{\psi_k}.\nonumber
    \end{align}
    Then 
    \begin{align}
        \trace(\rho \tL_i) &= 0 - \sum_{j:\lambda_{j}<\mu} {\braket{\psi_j|\partial_i\rho|\psi_j}} \nonumber\\
        &= 0 - \sum_{j:\lambda_{j}<\mu} \partial_i({\braket{\psi_j|\rho|\psi_j}}) = -\partial_i p, \\
    \widetilde J_{ij} &:=  \trace(\widetilde L_i\partial_j\rho)   = 
    \frac{1}{2}\trace(\rho\{\widetilde L_i,L_j\}) 
    \nonumber\\
    &= J_{ij} - \frac{1}{2}\trace(\rho_\perp\{\Pi_\perp L_i \Pi_\perp, \Pi_\perp L_j \Pi_\perp\} ). 
    \end{align}
    where in the second line we use ${\braket{\partial_i\psi_j|\rho|\psi_j}} + {\braket{\psi_j|\rho|\partial_i\psi_j}} = \lambda_j (\partial_j \braket{\psi_j|\psi_j}) = 0$. 
    Furthermore, 
    we have $\partial_j \lambda_j = 0$, $\partial_i p = 0$ and $J^p = 0$. Moreover, the SLDs of the post-selected state $\frac{\rho_\perp}{\trace(\rho_\perp)}$ on subspace $\Pi_\perp$ are $\Pi_\perp L_i \Pi_\perp$. They imply 
    \begin{equation}
        \trace(\rho \tL_i) = 0,\quad \tJ = \trace(\widetilde L_i\partial_j\rho) = J - p J^\perp. 
    \end{equation}
    
    Now we are ready to construct the locally unbiased estimators, defined through deviation observables 
    \begin{equation}
        \label{eq:def-X-low-rank}
        X_i = \sum_j (\tJ^{-1})_{ij} \tL_{j}. 
    \end{equation} 
    The locally unbiasedness condition is satisfied because 
    \begin{equation}
        \trace(\rho X_i) = 0,\quad \tJ_{ij} = \trace(X_i\partial_j\rho) = \delta_{ij}. 
    \end{equation}
    We can then evaluate $V$ using \eqref{eq:formula}. The second term is zero because $\trace(\tL_i) = \sum_{j:\lambda_j \geq \mu} \frac{\braket{\psi_j|\partial_i\rho|\psi_j}}{\lambda_j} = \sum_{j:\lambda_j \geq \mu}  \partial_i \braket{\psi_j|\psi_j} = 0$. Therefore, 
    \begin{equation}
    \label{eq:V-J-2}
        \tJ V \tJ = \frac{d+1}{d+2} ( \tK + 2 \tJ),
    \end{equation}
    where $\tK_{ij} := \trace(\tL_i\tL_j)$ and we use $\tJ_{ij} = \frac{1}{2}\trace(\rho\{\tL_i,\tL_j\})$ which holds because $\trace(\rho\{\tL_i,\Pi_\perp L_j \Pi_\perp\}) = 0$. For any $\vv \in \bR^m$, 
    \begin{multline*}
    \vv^T \tK \vv 
        = 4 \sum_{jk:\lambda_{j}\geq \mu \text{~or~} \lambda_{k}\geq \mu} \frac{\abs{\braket{\psi_j|\sum_i v_i \partial_i \rho|\psi_j}}^2}{(\lambda_j+\lambda_k)^2} \\
        \leq \frac{4}{\mu} \sum_{jk:\lambda_{j}\geq \mu \text{~or~} \lambda_{k}\geq \mu} \frac{\abs{\braket{\psi_j|\sum_i v_i \partial_i \rho|\psi_j}}^2}{\lambda_j+\lambda_k} = \frac{2}{\mu} \vv^T \tJ\vv. 
    \end{multline*}
    Therefore, $\tK \preceq \frac{2}{\mu} \tJ$. Using \eqref{eq:V-J-2}, we have 
    \begin{equation}
        V \preceq \frac{d+1}{d+2} \left(\frac{2}{\mu} + 2 \right)\tJ^{-1}. 
    \end{equation}
    Condition (2) of the definition of {approximately low-rank well-conditioned states} guarantees $\tJ \succeq c J$, which then implies \eqref{eq:V-J-3}. 
\end{proof}

\thmref{thm:unitary} is a generalization of \thmref{thm:pure}---\thmref{thm:pure} is directly implied from \thmref{thm:unitary} when taking $\mu = c = 1$ in \thmref{thm:unitary}. The key technique to prove \thmref{thm:unitary} is to choose the deviation observables wisely by excluding their components in subspace $\Pi_\perp$, which not only largely preserves the original QFIM but also makes upper bounding $\tK$ via the QFIM possible. 

Below, we consider the general case where unknown parameters can be encoded arbitrarily in the quantum state. We prove in this case, randomized measurements are still near-optimal for {approximately low-rank well-conditioned states}. 
\begin{theorem}
\label{thm:low-rank}
Consider a $(\mu,c)$-approximately low-rank well-conditioned state $\rho_\theta$. Let $M$ be a POVM given by a $3$-design. We have 
    \begin{equation}
        I(M) \succeq \frac{d+2}{d+1}\cdot \frac{\mu c^2}{2\mu c + 5}  J. 
    \end{equation}    
\end{theorem}
\begin{proof}
We prove the theorem by finding locally unbiased estimators such that 
    \begin{equation}
    \label{eq:V-J-4}
        V = V(M,\htheta) \preceq \frac{d+1}{d+2}\cdot\left(\frac{2}{c}+\frac{5}{\mu c^2}\right) J^{-1}.
    \end{equation}
Similar to the proof of \thmref{thm:unitary}, we adopt the local shadow estimators as our locally unbiased estimator, and use the unnormalized state $\rho_\perp := \Pi_\perp \rho \Pi_\perp$ for notation simplicity. 

Let $\tL_i$ be the SLDs of $\rho$. We define 
    \begin{align}
        \tL_i &= L_i - \Pi_\perp L_i \Pi_\perp + \frac{\partial_i p}{1-p} \Pi\\
        &= \sum_{jk:\lambda_{j} \text{~or~} \lambda_{k}\geq \mu} \bigg(\! \frac{2 \braket{\psi_j|\partial_i\rho|\psi_k}}{\lambda_j+\lambda_k} + \frac{\partial_i p}{1-p} \delta_{jk}\!\bigg)\ket{\psi_j}\bra{\psi_k}.\nonumber
    \end{align}
    Then 
    \begin{align}
        \trace(\rho \tL_i) &= 0 -\partial_i p + \partial_i p = 0, \\
    \widetilde J_{ij} &:=  \trace(\widetilde L_i\partial_j\rho) = \frac{1}{2} \trace(\rho\{\tL_i,L_j\}) \nonumber\\
    &= J_{ij} - p J^\perp_{ij} - J^p_{ij}. \label{eq:tildeJ-equality}
    \end{align}
    To verify \eqref{eq:tildeJ-equality} holds, we note that 
    \begin{align*}
        &J^\perp_{ii'} = 2 \! \sum_{\substack{jk,\,\lambda_{j,k} < \mu, \\\lambda_j+\lambda_k>0}} \! \frac{\Re[\braket{\psi_j|\partial_i(\rho/p)|\psi_k}\braket{\psi_k|\partial_{i'}(\rho/p)|\psi_j}]}{(\lambda_j+\lambda_k)/p}\\
        &= \frac{2}{p} \! \sum_{\substack{jk,\,\lambda_{j,k} < \mu, \\\lambda_j+\lambda_k>0}} \! \frac{\Re[\braket{\psi_j| \partial_i\rho |\psi_k}\braket{\psi_k| \partial_i\rho |\psi_j}]}{\lambda_j+\lambda_k} \! - \! \frac{1-p}{p}  J^p_{ii'}, \\
    \end{align*}
    where in the first equality we use the definition of QFIM and $\Pi_\perp\partial_i (\rho_\perp/p)\Pi_\perp = \Pi_\perp\partial_i (\rho/p)\Pi_\perp$. The second equality follows by direction calculations. Furthermore,    
    \begin{align*}
        \tJ_{ii'} &= 2 \! \sum_{jk,\,\lambda_{j} \text{~or~}\lambda_{k}\geq \mu} \! \frac{\Re[\braket{\psi_j| \partial_i\rho |\psi_k}\braket{\psi_k| \partial_i\rho |\psi_j}]}{\lambda_j+\lambda_k} \! - \! p  J^p_{ii'}.
    \end{align*}
    Combining the two equations above, we have $\tJ + p J^\perp + J^p = J$, which indicates \eqref{eq:tildeJ-equality}. 
    
    Let the deviation observables be $X_i = \sum_j (\tJ^{-1})_{ij} \tL_{j}$. The locally unbiasedness condition is satisfied because 
    \begin{equation}
        \trace(\rho X_i) = 0,\quad \tJ_{ij} = \trace(X_i\partial_j\rho) = \delta_{ij}. 
    \end{equation}
    Furthermore, let $\tK_{ij} = \trace(\tL_i\tL_j)$, we have, from \eqref{eq:formula},
    \begin{equation}
        \tJ V \tJ \preceq \frac{d+1}{d+2} (\tK + 2\tJ). 
    \end{equation}
    Note that the second term in \eqref{eq:formula} is negative semidefinite and thus can be ignored because it can only decrease the right-hand side. As proved in \appref{app:K-bound}, 
    \begin{equation}
    \label{eq:K-bound}
        \tK \preceq \frac{5}{\mu} J. 
    \end{equation}
    Condition (2) of the definition of {approximately low-rank well-conditioned states} guarantees $\tJ \succeq c J$, which then implies 
    \begin{align}
       V &\preceq \frac{d+1}{d+2} \left(\frac{5}{c\mu} \tJ^{-1} + 2\tJ^{-1}\right) \nonumber\\ &\preceq \frac{d+1}{d+2} \left(\frac{5}{\mu c^2} + \frac{2}{c}\right) J^{-1}. 
    \end{align}
\end{proof}

\thmref{thm:low-rank} indicates that for any {approximately low-rank well-conditioned state}, which includes the cases where the eigenvalues of the quantum states depend on $\theta$, randomized measurements can still be near-optimal. 

\section{Full-parameter mixed states}
\label{sec:full}

As discussed in \secref{sec:near-opt}, near-optimal measurements do not exist for general mixed states. In these cases, it is meaningful to consider weak near-optimality (\defref{def:weak}). Here we provide several families of quantum states where near-optimal measurements generally do not exist, and $3$-design randomized measurements are shown to be weakly near-optimal. The states discussed below are called \emph{rank-$r$ well-conditioned states}, which represent states which have $r$ non-zero eigenvalues and the ratio between the largest and the smallest non-zero eigenvalues are upper bounded by a constant. 

\begin{theorem}
\label{thm:full}
When $M$ is a POVM given by a $3$-design, it is weakly near-optimal for the following family of states: 
\begin{enumerate}[wide, labelwidth=!,itemindent=!,labelindent=0pt, leftmargin=0em, label={(\arabic*)}, parsep=0pt]
    \item Rank-$r$ well-conditioned states, whose number of unknown parameters is maximal, i.e., $m = r(2d-r) - 1$. 
    \item Rank-$r$ well-conditioned states $\rho_\theta$ that satisfies $\Pi \partial_i \rho_\theta\Pi = \partial_i \rho_\theta$ for all $i$ where $\Pi$ is the projection onto the support of $\rho_\theta$, whose number of unknown parameters is maximal, i.e., $m = r^2 - 1$. 
    \item Rank-$r$ well-conditioned states $\rho_\theta$ that satisfies $\Pi \partial_i \rho_\theta\Pi = 0$ for all $i$ where $\Pi$ is the projection onto the support of $\rho_\theta$, whose number of unknown parameters is maximal, i.e, $m = 2r(d-r)$.
\end{enumerate}    
\end{theorem}

Case (1) is the standard case of estimating a full-parameter rank-$r$ quantum state. Case (2) describes the situation where randomized measurements act on a larger Hilbert space than the Hilbert space $\rho_\theta$ lives in, which is common in situations where $\rho_\theta$ is restricted to a small subspace of the actual physical system. Case (3) describes the situation where the unknown parameters are encoded in the derivatives of $\rho_\theta$ which takes $\rho_\theta$ outside its support, with a typical example being pure states. A non-trivial example would be $\rho_\theta = U_\theta \rho_0 U_\theta^\dagger$ where $\theta$ is encoded in some unitary rotation $U_\theta$ and $\rho_0$ is the rank-$r$ maximally mixed state $\Pi/r$ independent of $\theta$. Note that here the rank $r$ can scale non-trivially with respect to the system dimension $d$ and the weak near-optimality in the three cases above does not imply one another. 

\begin{proof}[Proof~of~{\thmref{thm:full}}]
Here we consider rank-$r$ well-conditioned states. In particular, we assume the condition number inside the support of the state 
\begin{equation}
    \kappa = \frac{\max_i\lambda_i}{\min_{i:\lambda_i>0}\lambda_i}
\end{equation}
is a constant. Note that $\min_{i:\lambda_i>0}\lambda_i \geq 1/(r\kappa)$. 

To show weak near-optimality, we will find locally unbiased estimators such that $\trace(JV)$, where $V$ is the MSEM with respect to it, is at most a constant factor away from its minimum value.
Again, we adopt local shadow estimators (\eqref{eq:theta-estimator}) for the proof and use \eqref{eq:formula} to evaluate $\trace(JV)$. The deviation observables $\{X_i\}_i$ are defined by 
\begin{equation}
        X_i = \sum_j (J^{-1})_{ij} L_j,
\end{equation}
where $L_i$ are the SLDs. Let
\begin{align}
    K_{ij} & := \trace(L_iL_j) \nonumber\\
    &= 4\!\sum_{k\ell:\lambda_{k}+\lambda_{\ell}>0} \!\frac{\Re[\bra{\psi_k}\partial_i\rho_\theta\ket{\psi_\ell}\bra{\psi_\ell}\partial_j\rho_\theta\ket{\psi_k}]}{(\lambda_k+\lambda_\ell)^2}.
\end{align}
For any $\vv \in \bR^{m}$, 
\begin{multline*}
    \vv^T K \vv = 4\!\sum_{k\ell:\lambda_{k}+\lambda_{\ell}>0} \!\frac{\abs{\bra{\psi_k}\sum_i v_i\partial_i\rho_\theta\ket{\psi_\ell}}^2}{(\lambda_k+\lambda_\ell)^2} \\ 
    \leq 4r\kappa\!\sum_{k\ell:\lambda_{k}+\lambda_{\ell}>0} \!\frac{\abs{\bra{\psi_k}\sum_i v_i\partial_i\rho_\theta\ket{\psi_\ell}}^2}{\lambda_k+\lambda_\ell} = 2r\kappa \vv^T J \vv,
\end{multline*}
which implies $K \preceq 2r\kappa J$. From \eqref{eq:formula} and the CR bound $I(M)^{-1} \preceq V$, we have 
\begin{equation}
    J V J \preceq \frac{(d+1)}{d+2}(2 J + K) \preceq  \frac{2(d+1)(1+r\kappa)}{d+2}J,
\end{equation}
which implies 
\begin{equation}
    \trace(J I(M)^{-1}) \leq \trace(J V) \leq \frac{2(d+1)(1+r\kappa)m}{d+2}, \label{eq:upper}
\end{equation}
and, if we multiply both sides by $J^{-1/2}$ on the left and right and take the inverse, 
\begin{equation}
    I(M) \succeq   \frac{d+2}{2(d+1)(1+r\kappa)} J. 
\end{equation}

Below we use the GM bound (\eqref{eq:GM}~\cite{gill2000state}, see also \appref{app:GM}), to put upper bounds on $\trace(J^{-1} I(M))$ put lower bounds on $\trace(J I(\tM)^{-1})$ for general $\tM$ that match the scalings of the upper bounds in \eqref{eq:upper} in different cases.

When $\rho_\theta$ belongs to case (1), i.e., is a rank-$r$ well-conditioned state with maximal number of parameters. We have 
\begin{equation}
    m = r^2-1 + 2(d-r)r = r(2d-r) - 1 = \Theta(rd),
\end{equation}
where $r^2 - 1$ corresponds to the degrees of freedom within the support of $\rho_\theta$ and $2(d-r)r$ corresponds to the degrees of freedom outside the support of $\rho_\theta$. In this case, the upper bound (\eqref{eq:upper}) and the lower bound (\eqref{eq:lower}) on $\trace(J I(M)^{-1})$ are both $\Theta(r^2d)$.
Also note that $\Tr(J^{-1}I(M))\le d-1$ by GM bound.
These indicate the weak near-optimality of the randomized measurement $M$. 

In case (2), both $\rho_\theta$ and its derivatives $\partial_i \rho_\theta$ are restricted to the $r$-dimensional support of $\rho_\theta$. $m = r^2 - 1$. As shown in \appref{app:GM}, the GM bound can be tightened into 
\begin{equation}
    \trace(J^{-1}I(\tM)) \leq r-1,
\Rightarrow\,
    \trace(J I(\tM)^{-1}) \geq \frac{m^2}{r-1}. 
    \label{eq:lower-2}
\end{equation}
In this case, the upper bound (\eqref{eq:upper}) and the lower bound (\eqref{eq:lower-2}) on $\trace(J I(M)^{-1})$ are both $\Theta(r^3)$, indicating the weak near-optimality of the randomized measurement $M$.

In case (3), the derivatives $\partial_i \rho_\theta$ are restricted to the off-diagonal block orthogonal to the support of $\rho_\theta$. $m = 2r(d-r)$. As shown in \appref{app:GM}, the GM bound can be tightened into 
\begin{equation}
    \trace(J^{-1}I(\tM)) \leq d-r,
\,\Rightarrow\,
    \trace(J I(\tM)^{-1}) \geq \frac{m^2}{d-r}. 
    \label{eq:lower-3}
\end{equation}
In this case, the upper bound (\eqref{eq:upper}) and the lower bound (\eqref{eq:lower-3}) on $\trace(J I(M)^{-1})$ are both $\Theta(r^2(d-r))$, indicating the weak near-optimality of the randomized measurement $M$.
\end{proof}

\thmref{thm:full} demonstrates three cases of mixed states where randomized measurements are weakly near-optimal. The requirement that the states must have the maximal number of unknown parameters can be slightly relaxed---the theorem still holds as long as the scaling of the actual $m$ is the same as the scaling of the maximal number. 

We remark that in the special case of full-parameter maximally mixed states ($r = d$, $\kappa = 1$), a direct computation of the CFIM~\cite{zhu2018universally} showed $2$-design measurements are sufficient to achieve $I(M) = \frac{J}{d+1}$ and thus weakly near-optimal. $3$-designs are unnecessary in this case. In general, however, it is unclear whether $2$-design measurements will be sufficient for estimating full-parameter rank-$r$ states.

\section{Examples}
\label{sec:example}

In this section, we provide three examples to demonstrate the power and limitation of randomized measurements in multi-parameter quantum metrology. We also explicitly compute the deviation observables in several cases, and illustrate the advantage of local shadow estimators in one numerical example of fidelity estimation.

\subsection{Fidelity Estimation: Pure States}\label{sec:fid_pure}

\subsubsection{Theory} 

Here we consider estimating the fidelity of a parameterized pure state to a fixed quantum state. According to \thmref{thm:pure}, randomized measurements are near-optimal in this case. Without loss of generality, we assume 
\begin{equation}\label{eq:pure_fid_fiducial}
    \ket{\psi_{\theta}} := \ket{\psi_{f,{ g}}} = \sqrt{f} \ket{\phi} + \sqrt{1-f}\ket{\phi^\perp_{ g}}, 
\end{equation}
where $f := \abs{\braket{\phi|\psi_{\theta}}}^2$ is the fidelity of $\ket{\psi_{\theta}}$ to a fixed state $\ket{\phi}$, and $\ket{\phi^\perp_{ g}}$ is orthogonal to $\ket{\phi}$. Specifically, we let $\theta_1 = f$ and $(\theta_2,\theta_3,\ldots,\theta_m) = { g}$.  $\ket{\phi^\perp_{ g}}$ depends on ${ g}$ but $\ket{\phi}$ is independent of it. Here we explicitly calculate the MSE of estimating $f$ using the randomized measurement and the local shadow estimator defined in the proof of \thmref{thm:pure}, which can be shown to be optimal. First, note that 
\begin{equation}
    J_{ff} = 4 \Re[\braket{\partial_f\psi|\partial_f\psi} - \braket{\psi|\partial_f\psi}\braket{\partial_f\psi|\psi}] = \frac{1}{f(1-f)}. 
\end{equation}
and, for any $i$, 
\begin{equation}
    J_{f { g}_i} = 4 \Re[\braket{\partial_f\psi|\partial_{{ g}_i}\psi} - \braket{\psi|\partial_f\psi}\braket{\partial_{{ g}_i}\psi|\psi}] = 0. 
\end{equation}
The QFIM is block-diagonal. From \eqref{eq:def-X}, the deviation observable should be taken as $X_f = \frac{L_{f}}{J_{ff}}$, i.e.,
\begin{align}\label{eq:dev_op_fid}
    X_f 
    &= 2f(1-f)\ket{\phi}\bra{\phi}-2f(1-f)\ket{\phi^\perp_{ g}}\bra{\phi^\perp_{ g}} 
    \nonumber\\
    &\quad  +(1-2f)\sqrt{f(1-f)} (\ket{\phi}\bra{\phi^\perp_{ g}}+\ket{\phi^\perp_{ g}}\bra{\phi}),
\end{align}
and, from \eqref{eq:formula}, the MSE for estimating $f$ is 
\begin{align}\label{eq:local_fid_var}
    V_{ff} &= \frac{d+1}{d+2} (\trace(X_f^2) + 2\braket{\psi|X_f^2|\psi})\nonumber\\
    &= \frac{4(d+1)}{d+2} f(1-f). 
\end{align}
Now we explicitly compare the performance of the randomized measurement and local shadow estimator with two other standard and widely used estimators. 
\begin{enumerate}[wide, labelwidth=!,itemindent=!,labelindent=0pt, leftmargin=0em, label={(\arabic*)}, parsep=0pt]
    \item Optimal measurement and estimator. The projective measurement $M^{\opt} = \{\ket{\phi}\bra{\phi},\id - \ket{\phi}\bra{\phi}\}$ is optimal for fidelity estimation. The probability distribution of the projective measurement is $\{f,1-f\}$ and then we have 
    \begin{equation}
        I(M^{\opt})_{ff} = \frac{1}{f(1-f)},\quad I(M^{\opt})_{f{ g}_i} = 0,\;\forall i.
    \end{equation}
    Then the minimum MSE $V_{ff} = (I(M^{\opt})^{-1})_{ff} = f(1-f) = (J^{-1})_{ff}$ is achieved by the corresponding optimal locally unbiased estimator (see \eqref{eq:opt}). We see above that randomized measurements can achieve the same MSE up to a constant factor $4(d+1)/(d+2)$, as also proved in \thmref{thm:pure}. In practice, randomized measurements can be more advantageous than the optimal binary measurement in situations where the latter is challenging to implement, e.g., when $\ket{\phi}$ is highly entangled. Also, the choice of randomized measurements do not depend on $\ket{\phi}$ while the projective measurement must.
    \item Standard shadow estimator. Another standard method~\cite{huang2020predicting} to estimate fidelity using randomized measurements is to use the standard shadow estimator which is also a locally unbiased estimator defined by 
    \begin{equation}\label{eq:standard_shadow_estimator}
        \hat f^{\rm std}(s) = \bra{\phi}\hrho(s)\ket{\phi},  
    \end{equation}
    where $\hrho(s)$ is defined in \eqref{eq:rho-estimator}. When the measurement is given by a $3$-design, the corresponding MSE is 
    \begin{equation}\label{eq:shadow_fid_var}
        V_{ff}^{\rm std} = \frac{2(d+1)(1+2f)}{d+2} - (1+f)^2. 
    \end{equation}
    In contrast to the local shadow estimator which is only unbiased in the vicinity of its true value, the standard shadow estimator is globally unbiased and independent of the state and the parameters. However, the MSE obtained is not as small as the MSE obtained using the local shadow estimator. We have, for any $f$ and $d$, $V_{ff}^{\rm std} \geq V_{ff}$, and the equality holds only when $f=1/2$ and $d=2$. 
    
    In particular, when the fidelity $f$ approaches $1$, which means the quantum state $\ket{\psi_\theta}$ is close to the target state $\ket{\phi}$, we have $V_{ff} \rightarrow 0$, while $V_{ff}^{\rm std} \rightarrow \frac{2(d-1)}{d+2}$ remains a positive constant. It indicates a fundamental advantage of local shadow estimators over standard shadow estimators. 
\end{enumerate}

\begin{figure*}
    \centering
    \includegraphics[width=\linewidth]{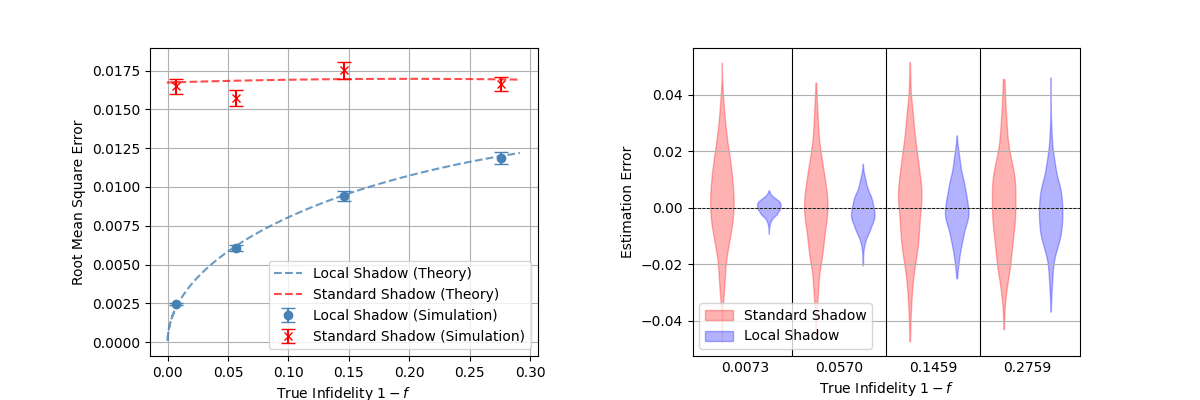}
    \caption{Simulation of a $3$-qubit pure state fidelity estimation task comparing standard and local shadow estimators. The target state $\ket{\phi}$ and four unknown states 
    with different infidelities all come from a $3$-parameter family of pure states described in Sec.~\ref{sec:numerics}. 
    For each unknown state, we build a dataset by uniformly sampling 100,000 random Clifford and measuring each for 10,000 shots. We then subsample $500$ batches from the dataset, each consists of $N=5000$ random Clifford and a single shot of measurement. 
    (Left.) Comparison of theoretical (square roots of \eqref{eq:local_fid_var} and \eqref{eq:shadow_fid_var}, divided by $\sqrt{N}$) and empirical (averaged over $500$ batches of subsamples) root mean squared error (RMSE) of estimators. Error bars depict one standard error estimated via $200$ round of bootstrap sampling~\cite{efron1994introduction}.
    (Right.) Violin plots (smoothed histograms~\cite{hintze1998violin}) of estimation errors from $500$ batches of subsamples for the four unknown states labeled by the corresponding infidelity.  
    The variance of our local shadow estimator decreases as the true infidelity decreases, whereas the variance of the standard shadow estimator remains nearly constant.
    }
    \label{fig:numerics}
\end{figure*}

\subsubsection{Numerical simulation}\label{sec:numerics}

To demonstrate the practical feasibility and advantages of our methods, we numerically implement a pipeline of pure state fidelity estimation using the local shadow estimator, and compare it with the standard shadow estimator~\cite{huang2020predicting} (\eqref{eq:standard_shadow_estimator}). Our simulation is conducted using Qiskit~\cite{qiskit}. The data and codes are available online (see the Data Availability section).

Our theoretical framework assumes knowing a fiducial state that is sufficiently close to the true unknown state, so that we can define the local shadow estimator (and deviation observable) with respect to that fiducial state. To apply this theory in practice, we adopt the following post-processing algorithm (after $N$ copies of states are measured using individual $3$-design measurements). One first obtains a coarse estimate of $(f^0,g^0)$ describing the unknown state (in the form of \eqref{eq:pure_fid_fiducial}) using, e.g., standard classical shadow~\cite{huang2020predicting}, and set it as the initial fiducial state. 
One can then use the initial fiducial state to compute the deviation observable (\eqref{eq:dev_op_fid}) and evaluate the local shadow estimator through
\begin{equation}
    f^1 =  f^0 + \frac{1}{N} \sum_s \trace\big( X_{f^0} \hrho(s)\big), 
\end{equation}
and update the estimator $\hat f$ from $f^0$ to $f^1$. Note that updating $\hat g$ is unnecessary because their influence on $\hat f$ is only second-order from the locally unbiasedness condition and can be ignored when the initial guess $(f^0,g^0)$ is sufficiently close to their true values. 
The above updating procedure is repeated and terminates at step $k$ when the stepwise change $\abs{f^k - f^{k-1}}$ converges below a small cutoff number, e.g., $10^{-6}$.
Remarkably, the entire algorithm operates by reusing the same set of measurement data, which means there is no sampling overhead of our methods compared to standard shadow estimator. Our protocol is summarized in Algorithm~\ref{alg:main}.

{\small
\begin{algorithm}[H]
\label{alg:main}
\DontPrintSemicolon
\SetAlgoLined
\SetKwInOut{Input}{Input}
\SetKwInOut{Output}{Output}
\Input{Target state $\ket{\phi}$ and unknown state $\ket{\psi}$}
\Output{Estimate of fidelity $f = \abs{\braket{\phi|\psi}}^2$}
\caption{Local Shadow Fidelity Estimation}
Repeat $N$ randomized measurements on $\ket\psi$'s\;
Obtain $N$ snapshots of classical shadow $\{\hat\rho(s)\}_s$\; 
Coarsely estimate $\hat f$ and $\hat g$ from snapshots $\{\hat\rho(s)\}_s$\;
\Repeat{$\hat\Delta < \mathrm{some~cutoff~value}$}{
    Compute $X_{\hat f}$ using $\hat f,\hat g$ via \eqref{eq:dev_op_fid}\;
    $\hat\Delta\gets\frac1N\sum_{s}\tr(X_{\hat f}\hat\rho(s))$\;
    $\hat f \gets \hat f + \hat\Delta$\;
}
\KwRet{$\hat f$}
\end{algorithm}
}
\vspace{0.1in}

As a concrete example, let us consider the following parameterized family of $n$-qubit pure states,
\begin{equation}
    \ket{\psi_{\varphi}} = \sqrt{\varphi_0}\ket{\Phi}+\sum_{i=1}^n\sqrt{\varphi_i}\ket{\Phi_i},
\end{equation}
where $\varphi_0=1-\sum_{i=1}^n \varphi_i$. $\ket\Phi=\frac{1}{\sqrt2}(\ket{0}^{\otimes n}+\ket{1}^{\otimes n})$ is the $n$-qubit GHZ state and $\ket{\Phi_i}=(\sigma_x)_i\ket{\Phi}$ is the GHZ with bit-flip on the $i$-th qubit. 
We choose $n=3$ and choose the target state $\ket{\phi}$ from fidelity estimation to be the one with parameters $\varphi_1=\varphi_2=\varphi_3=0.075$.
(Note this target state is not a stabilizer state and its projector cannot be implemented via Clifford gates.)
We estimate the fidelity of four unknown states to $\ket{\phi}$ where $\varphi_1=\varphi_2=\varphi_3 = 0.10,0.15,0.20,0.25$, yielding the true infidelities $1-f = 0.0073, 0.0570,0.1459,0.2759$, respectively.

We implement Algorithm~\ref{alg:main} to estimate fidelity from $N=5000$ randomized measurements and compare it to the standard shadow estimator (\eqref{eq:standard_shadow_estimator}). Specifically, we obtain the initial coarse estimates $\{\hat \varphi_i\}_{i=1}^3$ by computing the overlap between $\ket{\psi_\varphi}$ and $\ket{\Phi_i}$ using the standard shadow estimator
\begin{equation}
    \hat \varphi_i = \frac{1}{N}\sum_s \bra{\Phi_i}\hrho(s)\ket{\Phi_i}. 
\end{equation}
We then compute the initial $\hat f = f^0 = \abs{\braket{\phi|\psi_{\hat{\varphi}}}}^2$ and start to repetitively compute the local shadow estimators and update $\hat f$ until the stepwise difference is below $10^{-6}$. We compare the performance of Algorithm~\ref{alg:main} to the standard shadow estimator $\hat f^{\rm std} = \frac{1}{N}\sum_s \bra{\phi}\hrho(s)\ket{\phi}$. 
Both methods use exactly the same amount of experimental data and the difference is Algorithm~\ref{alg:main} adopts an additional post-processing procedure.

Our simulation results are reported in Fig.~\ref{fig:numerics}.
To demonstrate and compare the performances of both estimators, we create a large dataset containing many random Clifford and measurement shots, and then subsampling 500 batches of data to construct estimators and estimate their variance. Note that subsampling is a well-established method to assess the estimator variance while avoid the computational burden of full re-sampling~\cite{efron1994introduction}. See the caption of Fig.~\ref{fig:numerics} for more details. 
We also plot our theoretic predictions of the root mean squared error (RMSE, i.e., square root of the MSE) given by $\frac{1}{\sqrt{N}} \times $ the square roots of the single-shot MSE in \eqref{eq:local_fid_var} and \eqref{eq:shadow_fid_var}. The factor of $\frac{1}{\sqrt{N}}$ arises because the snapshots $\hat\rho(s)$ are i.i.d. For all values of fidelity simulated, the theoretical predictions agree with the empirical RMSEs within the error bars. Notably, though the MSE formula in \eqref{eq:local_fid_var} was derived in the context of local estimation, it also accurately captures the MSE for our Algorithm~\ref{alg:main} where no prior knowledge of $f$ is provided. The advantage of local shadow estimators over standard shadow estimators, as predicted by \eqref{eq:local_fid_var} and \eqref{eq:shadow_fid_var} is then established numerically. 

Regarding the run time of Algorithm~\ref{alg:main}, we observe that in most cases no more than $5$ repetitions are sufficient for the stepwise change of $\hat f$ to converge below $10^{-6}$.

\subsection{Fidelity Estimation: Mixed States}
\label{sec:no-go}

In contrast to the previous section, here we will show a no-go example where randomized measurements cannot be near-optimal---fidelity estimation for mixed states. Specifically, consider the following single-parameter family of $d$-dimensional quantum states
\begin{equation}\label{eq:1_par_mix_state}
    \rho_f := f\ketbra{\phi}{\phi} + (1-f)~\frac{\id-\ketbra{\phi}{\phi}}{d-1},
\end{equation}
for some fixed pure state $\ket{\phi}$. $\rho_f$ can be understood as $\ket{\phi}$ going through a depolarizing channel, and our goal is to estimate the fidelity $f := \bra{\phi}\rho\ket{\phi}$. Below we consider the regime  where $1 - f\le d^{-2}$, which means the mixed state $\rho_f$ is very close to the target state $\ket{\phi}$.
The QFI for $\rho_f$ is easily seen to be
\begin{equation}
    J_{ff} = \frac{1}{f(1-f)}, 
\end{equation}
We now show that randomized measurements are not near-optimal for this task. In particular, this holds for any $d$ (even when $d = 2$). Since weak near-optimality is equivalent to near-optimality in single-parameter metrology tasks, randomized measurements are also not weakly near-optimal here. 

Let us first clarify our definition of randomized measurements here. 
The results we present in \secref{sec:pure}--\secref{sec:full} are sufficient conditions for randomized measurements to be near-optimal or weakly near-optimal. For all these results, it suffices to take the POVM as any $3$-design.
To derive no-go results, however, we need to specify the exact choice of the POVMs. Below we prove the no-go results when the POVM exactly forms a (complex projective) Haar measure, and when it is a specific $3$-design.

\begin{enumerate}[wide, labelwidth=!,itemindent=!,labelindent=0pt, leftmargin=0em, label={(\arabic*)}, parsep=0pt]
\item Haar measure. Let the POVM element be $M^{\text{Haar}}=\{d \,\mathrm{d}\mu_s\ketbra{s}{s}\}_s$, where $\{\mathrm{d}\mu_s\ketbra{s}{s}\}_s$ is the (complex projective) Haar measure for ${ d}$-dimensional pure states~\cite{watrous2018theory}. Note that the dimension $d$ should be distinguished from the differential symbol $\mathrm{d}$. The probability density $p_s := \trace\big(\frac{\mathrm{d}M_s^{\text{Haar}}}{\mathrm{d}\mu_s}\rho_f\big)$ is 
\begin{align}
    p_s 
    &= { d}\left(f\abs{\braket{s|\phi}}^2 + (1-f)\,\frac{1-\abs{\braket{s|\phi}}^2}{{ d}-1}\right).
\end{align}
Define $\tau_s:=\abs{\braket{s|\phi}}^2$.
The CFI can be calculated to be
\begin{align}
    I(M^{\text{Haar}})_{ff} 
    &= \int \mathrm{d}\mu_s \frac{(\partial_f p_s)^2}{p_s}\nonumber\\
    &= \frac{{ d}}{({ d}-1)^2}\int \mathrm{d}\mu_s  \frac{(1-\tau_s{ d})^2}{f\tau_s + (1-f)\frac{1-\tau_s}{{ d}-1}}. 
\end{align}
Dividing the integral over $s$ into two parts depending on whether $\tau_s<\sqrt{1-f}$ holds or not, we have the following upper bound,
\begin{align}
    &I(M^{\text{Haar}})_{ff} \le \frac{{ d}}{({ d}-1)^2}\left(\int_{\tau_s>\sqrt{1-f}}\mathrm{d}\mu_s \frac{(1-\tau_s{ d})^2}{\sqrt{1-f}}\right.
    \nonumber\\&\qquad\qquad\qquad\left.+\int_{\tau_s\le\sqrt{1-f}}\mathrm{d}\mu_s  \frac{(1-\tau_s{ d})^2}{(1-f)\frac{1}{d-1}}\right)
    \nonumber\\&\quad\le\frac{{ d}}{{ d}-1}\left(\frac{{ d}-1}{\sqrt{1-f}}+\frac{\Pr[\tau_s\le\sqrt{1-f}]}{{1-f}}\right).
\end{align}
In the first inequality, we use $f\tau_s + (1-f)\frac{1-\tau_s}{d-1} \geq (1-f)\frac{1}{d-1}$ and, when $\tau_s \geq \sqrt{1-f}$ and $\sqrt{1-f} \leq d^{-1}$, $f\tau_s + (1-f)\frac{1-\tau_s}{d-1} \geq \sqrt{1-f}$. In the second inequality, we use the fact that $(1-\tau_s d)^2\le(d-1)^2$ and that, when $\tau_s \leq \sqrt{1-f} \leq d^{-1}$, $(1-\tau_s d)^2\le 1$.
In \appref{app:spherical}, we show that $\Pr[\tau_s\le\sqrt{1-f}]\le({ d}-1)\sqrt{1-f}$. Thus, we have
\begin{gather}\label{eq:fid_mixed}
    I(M^{\text{Haar}})_{ff} \le \frac{2d}{\sqrt{1-f}},
    \\
    \Rightarrow\, 
    \frac{I(M^{\text{Haar}})_{ff}}{J_{ff}} \leq 2d\sqrt{1-f}f.
\end{gather}
When taking $1 - f \rightarrow 0$, we have $\frac{I(M^{\text{Haar}})_{ff}}{J_{ff}} \rightarrow 0$. This completes the proof that $M^{\text{Haar}}$ is not near-optimal. 

\item $3$-design. 
For completeness, we now give another example that a specific POVM forming a $3$-design is not near-optimal for mixed state fidelity estimation.
Take \eqref{eq:1_par_mix_state} to be an single-qubit state.
\begin{equation}
    \rho_f = f\ketbra{\phi}{\phi} + (1-f)\ketbra{\phi^\perp}{\phi^\perp}.
\end{equation}
where $\ketbra{\phi^\perp}{\phi^\perp}:=\id-\ketbra{\phi}{\phi}$. 
Consider single-qubit random Pauli measurements 
\begin{equation}
    M^{\mathrm{P}}=\left\{\frac13\ketbra{x\pm}{x\!\pm\!},\frac13\ketbra{y\pm}{y\!\pm\!},\frac13\ketbra{z\pm}{z\!\pm\!}\right\}.
\end{equation}
Here $\ket{x\pm}$ denotes the $\pm 1$ eigenstate for Pauli X, similar for the others. $M^{\mathrm{P}}$ forms a (complex projective) $3$-design satisfying \eqref{eq:$3$-design}.
Now choose 
\begin{align}
    \ketbra{\phi}{\phi}:=\frac{1}{2}\left(\id + \frac{1}{\sqrt3}(\sigma_x+\sigma_y+\sigma_z)\right). 
\end{align}
In the Bloch sphere, $\ket{\phi}$ is located equidistant from $\ket{x+}$, $\ket{y+}$, and $\ket{z+}$. With some calculation, one can obtain the measurement outcome probability and the CFI as
\begin{gather}
    p_{a\pm} := \trace(\rho_f M_{a\pm}^{\mathrm{P}}) = \frac{1}{6}\left(1\pm\frac{2f-1}{\sqrt3}\right),\;\forall a\in\{x,y,z\},\nonumber\\
    I(M^\mathrm{P})_{ff} = \frac{2}{1+2f(1-f)},
\end{gather}
When we take $1-f \rightarrow 0$, we have $\frac{I(M^{\text{P}})_{ff}}{J_{ff}} \rightarrow 0$, thus $M^{\text{P}}$ is not near-optimal despite being a $3$-design.
As a final comment, if we rotate $M^{\text{P}}$ such that $\ketbra{\phi}{\phi}$ is proportional to one of the POVM elements, then the measurement yields a CFI of $\frac13J_{ff}$ and is thus near-optimal. 
However, such choices of measurement are tailored to the state, and thus should not be viewed as a generic randomized measurement for our purpose. 
\end{enumerate}

\subsection{Hamiltonian Estimation}
\label{sec:Hamt}

Finally, we provide one example of quantum \emph{channel estimation}, which goes beyond the previously discussed \emph{state estimation} problem. 

We will see that \thmref{thm:pure} and \thmref{thm:unitary} can be applied in a Hamiltonian estimation problem in the noiseless and the noisy case, respectively. 
Specifically, we consider the estimation of all Pauli operators in a Hamiltonian and show two near-optimal estimation protocols minimizing the WMSE with $W = \id$, one utilizing maximally entangled states and the other one utilizing random stabilizer states as input states. $3$-design randomized measurements are near-optimal in both the noiseless case (from \thmref{thm:pure}) and the noisy case (from \thmref{thm:unitary}). 

\subsubsection{Noiseless case}\label{sec:ham_noiseless}

We consider a Hamiltonian estimation problem in an $n$-qubit system, where 
\begin{equation}
    U_\theta = \exp\bigg({-i \sum_{k=1}^{d^2-1} \theta_k P_k}\bigg),
    \label{eq:hamt-est}
\end{equation}
where $d = 2^n$ and $\{P_k\}_{k=1}^{d^2-1}$ are the set of multi-qubit Pauli operators $\{\id,\sigma_x,\sigma_y,\sigma_z\}^{\otimes n}$ excluding $\id$. Consider parameter estimation at $\theta = 0$ 
for pure input states\footnote{All functions we calculate in this section will be implicitly evaluated at $\theta = 0$, unless specified otherwise.}. Let the input state of the unitary evolution be $\ket{\psi}$, we have the QFIM $J(\psi)$ of the output state 
\begin{equation}
\ket{\psi_\theta} := U_\theta \ket{\psi}    
\end{equation}
equal to 
\begin{align}
J(\psi)_{jk}\! &= 4\Re[\braket{\partial_j\psi_\theta|\partial_k\psi_\theta} - \braket{\psi_\theta|\partial_j\psi_\theta}\braket{\partial_k\psi_\theta|\psi_\theta}] \nonumber\\
&= 4 ( \!\bra{\psi}\!\frac{1}{2}\{P_j,P_k\}\!\ket{\psi} - \bra{\psi}\!P_j\!\ket{\psi}\bra{\psi}\!P_k\!\ket{\psi} \!), \label{eq:qfi-hamt}
\end{align}
where we use $J(\star)$\footnote{This notation should be distinguished from $J[\star]$ which is the QFIM of the parameterized state $\star$ that we use later.} to denote the QFIM when the input state is $\star$. 

\begin{itemize}[wide, labelwidth=!,itemindent=!,labelindent=0pt, leftmargin=0em, parsep=0pt]
\item \textbf{Figure of merit}
\end{itemize}

Different choices of input states $\ket{\psi}$ lead to different performances of Hamiltonian estimation. To evaluate their performances, we will use the trace of the inverse of the CFIM as a function of the input state $\psi$ and the measurement $M$, i.e., 
\begin{equation}
    \trace(I(\psi,M)^{-1})
\end{equation}
as the figure of merit. The function $\trace(I(\psi,M)^{-1})$ can be interpreted as the minimum achievable WMSE (according to the CRB) with the cost matrix $W = \id$\footnote{
Here it is not suitable to choose the cost matrix $W = J(\psi)$ as in the definition of weak near-optimality, because we aim to compare different input states (for channel estimation) instead of different POVMs (for state estimation) and the cost matrix should not depend on the input state. $W = \id$ is the most natural choice as it assigns equal importance to all parameters.}, when the input state is $\ket{\psi}$ and the measurement is $M$. 
The above can be generalized to the case where we use an input state ensemble to estimate the Hamiltonian, instead of a single input state. Specifically, we can assume, with probability $p^{(i)}$, we pick $\ket{\psi^{(i)}}$ as an input state and $M^{(i)}$ as the corresponding POVM. Then we use $I(\{p^{(i)},\psi^{(i)},M^{(i)}\})$ 
to denote the corresponding CFIM. In particular, we have 
\begin{align}
\label{eq:cfim-cov}
    &\quad \;I(\{p^{(i)},\psi^{(i)},M^{(i)}\}) \\
    &= I\left(\sum_{i}p_i \ket{i}\bra{i}_{\rm R}\otimes \psi^{(i)},\sum_{i}p_i \ket{i}\bra{i}_{\rm R} \otimes M^{(i)}\right)\\
    &= \sum_i p^{(i)} I(\psi^{(i)},M^{(i)}),
\end{align}
where we introduce the fictitious reference system R to facilitate the calculation. 

Below, we consider two types of input state ensembles for the Hamiltonian estimation problem. We will see they are both near-optimal and are equivalent up to a constant factor. Here, we say an input state ensemble (along with a corresponding measurement) is \emph{near-optimal}, if it achieves the minimum WMSE optimized over all input state ensembles and measurements, i.e.,  
\begin{equation}
\label{eq:def-w}
    w:= \inf_{\{p^{(i)},\psi^{(i)},M^{(i)}\}}\trace(I(\{p^{(i)},\psi^{(i)},M^{(i)}\})^{-1})
\end{equation}
up to a constant factor. Note that it is sufficient to consider only pure input states here because of the convexity of the CFIM. In \eqref{eq:def-w}, we also allow $\psi^{(i)}$ to be an entangled state (and $M$ to be a correlated measurement) across the probe system and an arbitrarily large ancillary system as explained below. 

\begin{itemize}[wide, labelwidth=!,itemindent=!,labelindent=0pt, leftmargin=0em, parsep=0pt]
\item \textbf{Maximally entangled states}
\end{itemize}

In the first scenario, we make use of an ancillary system $\mathrm{A}$, in which case the input state can be an entangled state between the original probe system $\mathrm{P}$ and the ancillary system. The Hamiltonian evolution acts only on the probe system, and the measurement can act on the joint system. Specifically, we choose the maximally entangled state as the input state, i.e., 
\begin{equation}
\label{eq:mm-state}
    \ket{\psi^{\mathrm{ME}}_{\mathrm{P}\mathrm{A}}} = \frac{1}{\sqrt{d}} \sum_{k=0}^{d-1} \ket{k}_{\mathrm P}\ket{k}_{\mathrm A},
\end{equation}
where $\{\ket{k}\}_{k=0}^{d-1}$ is the computational basis. The QFIM can be calculated through \eqref{eq:qfi-hamt} where $P_j$ are interpreted as $(P_j)_{\mathrm P} \otimes \id_{\mathrm A}$, and we have 
\begin{equation}
    J(\psi^{\mathrm{ME}}) = 4\id,\;\text{and}\; \trace(J(\psi^{\mathrm{ME}})^{-1}) = \frac{d^2-1}{4}. 
\end{equation}
The deviation observable for each parameter is
\begin{equation}
\label{eq:X-ME}
    X_j(\psi^{\mathrm{ME}};\theta) = \frac{i}{2} \big[\ket{\psi^{\mathrm{ME}}_\theta}\!\bra{\psi^{\mathrm{ME}}_\theta}, P_j\big], 
\end{equation}
which induces the local shadow estimator that achieve the QFIM up to a constant factor via $3$-design measurements. When $\theta$ is sufficiently close to zero, in order to estimate $\theta_j$, it is sufficient to use $X_j(\psi^{\mathrm{ME}};(0,\cdots,\theta_j,\cdots, 0)) \approx X_j(\psi^{\mathrm{ME}};0) + \theta_j \frac{d+1}{2d}   \big[\left(P_j \ket{\psi^{\mathrm{ME}}}\!\bra{\psi^{\mathrm{ME}}} - \ket{\psi^{\mathrm{ME}}}\!\bra{\psi^{\mathrm{ME}}} P_j\right), P_j\big] $ as the deviation observable to estimate $\theta_j$, where we set $\theta_{k\neq j} = 0$ because their influence on the local shadow estimator is second-order and we also ignore the  $O(\theta_j^2)$ term. 
The corresponding local shadow estimator $\trace(X_j \hat\rho(s))$ is then efficiently computable when the $3$-design measurement bases $\ket{s}$ are random stabilizer states using the property that $\ket{\psi^{\mathrm{ME}}}$ is also a stabilizer state. 

$\ket{\psi^{\mathrm{ME}}}$ minimizes $\trace(J(\psi)^{-1})$, because for any $\psi$, 
\begin{equation}
    \trace(J(\psi)^{-1}) \geq (d^2-1)^2 \trace(J(\psi))^{-1} \geq \frac{d^2-1}{4}. 
\end{equation}
The above equalities hold when $\psi = \psi^{\rm ME}$. The first inequality above is Cauchy--Schwarz and the second is due to $J(\psi) \preceq Q(\psi)$ and 
\begin{equation}
    \label{eq:opt-wmse}
    \trace(J(\psi)) \leq \trace(Q(\psi)) = 4(d^2-1),
\end{equation} 
for any input state $\psi$, where $Q_{jk} := 4 \!\bra{\psi}\!\frac{1}{2}\{P_j,P_k\}\!\ket{\psi}$. This implies 
\begin{equation}
\label{eq:w-lower}
    w \geq \frac{d^2-1}{4},
\end{equation}
because for any $\epsilon > 0$, there exists an input state ensemble $\{p^{(i)},\psi^{(i)},M^{(i)}\}$ such that
\begin{align}
    w + \epsilon 
    &\geq \trace(I(\{p^{(i)},\psi^{(i)},M^{(i)}\})^{-1}) \nonumber \\
    &= \trace\bigg(\Big(\sum_i p^{(i)} I(\psi^{(i)},M^{(i)})\Big)^{-1}\bigg) \nonumber \\ 
    &\geq (d^2-1)^2 \Big(\sum_i p^{(i)} \trace(I(\psi^{(i)},M^{(i)}))\Big)^{-1} \nonumber \\
    &\geq (d^2-1)^2  \Big(\sum_i p^{(i)} \trace(J(\psi^{(i)}))\Big)^{-1} \nonumber \\
    &\geq \frac{d^2-1}{4} = \trace(J(\psi^{\rm ME})^{-1}), 
\end{align}
where we use \eqref{eq:cfim-cov} in the second line, the Cauchy-Schwarz in the third line, the CRB $J(\psi^{(i)})^{-1} \preceq I(\psi^{(i)},M^{(i)})^{-1}$ in the fourth line and \eqref{eq:opt-wmse} in the last line. 
Thus, $\psi^{\mathrm{ME}}$ is near-optimal for the Hamiltonian estimation problem. For any measurement $M$ given by a $3$-design, 
\begin{equation}
    \trace(I(\psi^{\mathrm{ME}},M)^{-1}) 
    \leq 4 \trace(J(\psi^{\rm ME})^{-1}) \leq 4w,\label{eq:w1}
\end{equation}
where we use \thmref{thm:pure} to show the first inequality.

\begin{itemize}[wide, labelwidth=!,itemindent=!,labelindent=0pt, leftmargin=0em, parsep=0pt]
\item \textbf{Random stabilizer states}
\end{itemize}

In the second scenario, we consider ancilla-free scenarios, which is more favorable in practice because the ancillary system is not always available. An input state ensemble is necessary here because with a single input state, at most $d-1$ number of parameters can be estimated. Here we choose 
\begin{equation}
\label{eq:clifford-input}
    p^{\mathrm{C},(i)} = \frac{1}{\abs{\mU}},\quad \ket{\psi^{\mathrm{C},(i)}} = U^{(i)}\ket{0^{\otimes n}},
\end{equation}
where $U^{(i)}$ is an element from the random Clifford unitary ensemble $\mU$. Let the joint state on the fictitious reference system R and the probe system P be
\begin{equation}
    \rho^{\mathrm{C}}_{\rm{RP}} = \frac{1}{\abs{\mU}} \sum_{i} \ket{i}\bra{i}_{\rm R}\otimes \psi^{\mathrm{C},(i)}. 
\end{equation}
Then 
\begin{align}
    &\quad\; J(\rho^{\mathrm{C}}_{\rm{RP}}) = \sum_i p^{\mathrm{C},(i)} J(\psi^{\mathrm{C},(i)})_{jk}
    \\
    \label{eq:use_two_design}
    &= 4\left(\delta_{jk} - \frac{\delta_{jk}}{d+1}\right) = \frac{4d\delta_{jk}}{d+1},
\end{align}
where we use 
\begin{gather}
    \bE_i[\bra{\psi^{\mathrm{C},(i)}}\!\frac{1}{2}\{P_j,P_k\}\!\ket{\psi^{\mathrm{C},(i)}}] = \frac{1}{d}\trace\Big(\frac{1}{2}\{P_j,P_k\}\Big),\nonumber\\\bE_i[\bra{\psi^{\mathrm{C},(i)}}\! P_j\!\ket{\psi^{\mathrm{C},(i)}}\!\bra{\psi^{\mathrm{C},(i)}}\! P_k \!\ket{\psi^{\mathrm{C},(i)}}] = \frac{\trace(P_jP_k)}{d(d+1)}. \nonumber
\end{gather}
These are properties of random Clifford unitaries as a 1-design and a $2$-design. For any measurement $M$ given by a $3$-design, 
\begin{align}
    &\quad \trace(I(\{p^{\mathrm{C},(i)},\psi^{\mathrm{C},(i)},M\})^{-1}) \nonumber \\
    &\leq 4\trace\Big(\Big(\sum_i p^{\mathrm{C},(i)}  J(\psi^{\mathrm{C},(i)})\Big)^{-1}\Big)  \nonumber\\
    &= \frac{(d+1)}{d}(d^2-1) \leq \frac{4(d+1)}{d}w,\label{eq:w2}
\end{align}
where we use \eqref{eq:cfim-cov} and \thmref{thm:pure} to show the first inequality and \eqref{eq:w-lower} and \eqref{eq:use_two_design} in the second step. 

For input states chosen randomly from a state ensemble, we define the deviation observable for each parameter as 
\begin{align}
    X_j(\psi^{\mathrm{C},(i)};\theta) &= \sum_{k} (J(\rho^{\mathrm{C}}_{\rm{RP}})^{-1})_{jk} L_k(\psi^{\mathrm{C},(i)}) \\
    &= \frac{i(d+1)}{2d} \big[\ket{\psi^{\mathrm{C},(i)}_\theta}\!\bra{\psi^{\mathrm{C},(i)}_\theta}, P_j\big].\label{eq:X-C}
\end{align}
They induce the following local shadow estimator (at $\theta^0$) as a generalization of \eqref{eq:theta-estimator} to the case of random input states,
\begin{equation}
    \htheta_j(s;i,\theta^0) = \theta_i^0 + \trace\big( (X_j(\psi^{\mathrm{C},(i)};\theta) |_{\theta = \theta^0}) \hrho(s)\big). 
\end{equation}
It is a function of both $i$ and $\theta^0$ where the index $i$ of the choice of the input state that is known to the experimenter. The overall variance $V$ in this case is the average of the variance for different input states, which is given by
\begin{multline}
    (J_{\rm{RP}} V J_{\rm{RP}})_{jk} = \frac{d+1}{d+2}\bE_i\Big[  \trace(L^{\mathrm{C},(i)}_jL^{\mathrm{C},(i)}_k) \\  + \trace(\psi^{\mathrm{C},(i)}_\theta\{L^{\mathrm{C},(i)}_j,L^{\mathrm{C},(i)}_k\}) \Big]
    \\ = \frac{d+1}{4(d+2)}\bE_i[J(\psi^{\mathrm{C},(i)})] =  \frac{d+1}{4(d+2)} J_{\rm{RP}},
\end{multline}
where $J_{\rm{RP}}$ is the shorthand of the averaged QFIM $J(\rho^{\mathrm{C}}_{\rm{RP}})$. 
The above shows how the local shadow estimators achieve the near-optimal performance and the derivation follows exactly from the proof of \thmref{thm:pure} after taking random inputs into consideration.

To summarize, the discussion above (in particular, \eqref{eq:w1} and \eqref{eq:w2}) implies the following corollary of \thmref{thm:pure}: 
\begin{corollary}
    The maximally entangled state, or the random stabilizer state ensemble, together with any $3$-design measurement, forms a near-optimal protocol for the Hamiltonian estimation problem in \eqref{eq:hamt-est} at $\theta = 0$. 
\end{corollary}

\subsubsection{Noisy case}\label{sec:ham_noisy}

Using \thmref{thm:unitary}, the above discussions can be generalized to the situations where quantum noise affects the Hamiltonian estimation. We assume a Pauli noise channel, in which case the parameterized quantum channel acting on the probe system $\mathrm{P}$ becomes 
\begin{equation}
\mE_\theta(\cdot) = (1-q) U_\theta(\cdot)U_\theta^\dagger + \sum_{k = 1}^{d^2 - 1} q_k P_k U_\theta(\cdot)U_\theta^\dagger P_k, 
\end{equation}
where $q_k$ are the probabilities of Pauli error $P_k$ which sum up to the total noise rate $q$, and $q < 1/2$. We also define $q_0:=1-q$ and $P_0 = \id$ for future use. Again, we consider parameter estimation at $\theta = 0$. We will show that the maximally entangled state and the random stabilizer state ensemble can still achieve the optimal WMSE in the noiseless case (\eqref{eq:def-w}) up to a constant factor that depends only on $q$.

\begin{itemize}[wide, labelwidth=!,itemindent=!,labelindent=0pt, leftmargin=0em, parsep=0pt]
\item \textbf{Maximally entangled states}
\end{itemize}

In the first ancilla-assisted scenario, we have the maximally entangled state (\eqref{eq:mm-state}) as the input state. We assume no noise on the ancillary system. Then $(P_k \otimes \id)\ket{\psi^{\rm ME}}$ are orthogonal to each other for different $k$. (Below we omit $\otimes \id$ for simplicity.) Let $\Pi$ be the rank-one projector onto $\ket{\psi_\theta^{\rm ME}} := U_\theta\ket{\psi^{\rm ME}}$, which is the eigenstate of the output state $\mE_\theta(\psi^{\rm ME})$ with the largest eigenvalue, and the probability of post-selecting on subspace $\Pi_\perp$ is $\trace(\Pi_\perp \mE_\theta(\psi)) = q$. The QFIM after post-selecting on $\Pi_\perp$ satisfies 
\begin{align}
    &\quad\, J^\perp(\psi^{\rm ME}) := J\left[\sum_{k=1}^{d^2-1} \frac{q_k}{q} P_k U_\theta \psi^{\rm ME} U_\theta^\dagger P_k\right] \nonumber\\
    &\preceq \sum_{k=1}^{d^2-1} \frac{q_k}{q} J[P_k U_\theta \psi^{\rm ME} U_\theta^\dagger P_k] = J[\psi_\theta^{\rm ME}] = 4\id,\label{eq:noisy-bell-lower}
\end{align}
where we use the convexity of the QFIM in the first inequality, $J[P_k \psi_\theta^{\rm ME}  P_k] = J[\psi_\theta^{\rm ME}]$ in the second equality. 
Here, we use $J[\star]$ to denote the QFIM of state $\star$, and thus $J(\star)=J[\mathcal E_\theta(\star)]$ by definition.
(Previously, the dependence of $J$ on the parameterized states was implicit, as these states were always straightforwardly $\rho_\theta$ or $\psi_\theta$. However, in the current context involving multiple distinct parameterized states, it is necessary to explicitly indicate this dependence.) Meanwhile, the QFIM of the output state $\mE_\theta(\psi^{\rm ME})$ is 
\begin{align}
    J(\psi^{\rm ME})_{ij} &:= J\left[\sum_{k=0}^{d^2-1} q_k P_k U_\theta \psi^{\rm ME} U_\theta^\dagger P_k\right]_{ij} \nonumber\\
    &= 2\delta_{ij} \sum_{\substack{0\leq k \leq d^2-1,\\q_k + q_{k+i}>0}} \frac{\abs{q_{k+i} - q_k}^2}{q_k + q_{k+i}}.
\label{eq:noisy-bell}
\end{align} 
Note that the $+$ sign in $q_{k+i}$ should be interpreted in the following way: $q_{k+i}$ denotes the noise rate of the Pauli operator $P_{k+i} \propto \eta_{k,i}P_{k}P_{i}$. 
The derivation of \eqref{eq:noisy-bell} can be found in \appref{app:bell}. Then 
\begin{equation}
    J(\psi^{\rm ME})_{ij} \geq 4 \delta_{ij}\frac{\abs{q_{i} - q_0}^2}{q_0 + q_{i}} \geq 4 \delta_{ij}(1-2q)^2,\label{eq:noisy-bell-upper}
\end{equation}
where we use $q_i \leq q$ in the last step. Combining \eqref{eq:noisy-bell-lower} and \eqref{eq:noisy-bell-upper}, we have 
\begin{equation}
    J(\psi^{\rm ME}) - q J^{\perp}(\psi^{\rm ME}) \succeq \frac{(1-2q)^2-q}{(1-2q)^2} J(\psi^{\rm ME}).
\end{equation}
This means the noisy output state $\mE_\theta(\psi^{\rm ME})$ is a $(1-q,\frac{(1-2q)^2-q}{(1-2q)^2})$-approximately low-rank well-conditioned state, according to \defref{def:low-rank}. Applying \thmref{thm:unitary}, we know when $M$ is given by a $3$-design,  
\begin{align}
    I(\psi^{\rm ME},M) & \succeq \frac{d+2}{d+1}\frac{(1-q)^2(1-4q)}{2(2-q)(1-2q)^2} J(\psi^{\rm ME}) \nonumber\\
    &\succeq \frac{d+2}{d+1}\frac{(1-q)^2(1-4q)}{2(2-q)} (4\id), \label{eq:lower-bell}
\end{align}
where $J(\psi^{\rm ME})$ is the QFIM with the maximally entangled state as the input state for the noisy Hamiltonian estimation problem and $4\id$ is the QFIM of the maximally entangled input state in the noiseless case. 

We can also calculate the deviation observable for each parameter (\eqref{eq:def-X-low-rank}). Using 
\begin{gather}
    \tL_j(\psi_{\mathrm{ME}}) = \frac{2i(q_0-q_j)}{q_0+q_j} \big[\ket{\psi^{\mathrm{ME}}_\theta}\!\bra{\psi^{\mathrm{ME}}_\theta}, P_j \otimes \id_{\mathrm A}\big],\\
    \tJ_{jk}(\psi_{\mathrm{ME}}) =  \frac{4(q_0-q_j)^2}{(q_0+q_j)^2}\delta_{jk}. 
\end{gather}
we have 
\begin{equation}
\label{eq:X-ME-1}
    X_j(\psi_{\mathrm{ME}};\theta) = \frac{i(q_0+q_j)}{2(q_0-q_j)} \big[\ket{\psi^{\mathrm{ME}}_\theta}\!\bra{\psi^{\mathrm{ME}}_\theta}, P_j\big], 
\end{equation}
which is proportional to the deviation observable in the noiseless case (\eqref{eq:X-ME}). The simplification stems from the proof of \thmref{thm:unitary} where we show the sufficiency to achieve near-optimal estimation by excluding the $\Pi_\perp$ component in the SLD operators.  

\begin{itemize}[wide, labelwidth=!,itemindent=!,labelindent=0pt, leftmargin=0em, parsep=0pt]
\item \textbf{Random stabilizer states}
\end{itemize}

In the second ancilla-free scenario, we use the random stabilizer state ensemble (\eqref{eq:clifford-input}). We first consider the special case where the input state is $\ket{\psi^{\mathrm{C},(0)}} := \ket{0^{\otimes n}}$. Let $\ket{\psi^{\mathrm{C},(0)}_\theta} := U_\theta\ket{0^{\otimes n}}$. Let $\Pi$ be the rank-one projector onto $U_\theta\ket{0^{\otimes n}}$, and $\Pi^\perp = \id - \Pi$. The probability of post-selecting on subspace $\Pi_\perp$ is 
\begin{equation}
    \trace\big(\Pi_\perp \mE_\theta\big(\psi^{\mathrm{C},(0)}\big)\big) = \sum_{k:\tk \neq 0^n} q_k =: \tilde{q} \leq q,
\end{equation}
where we define $\tk$ to be an $n$-bit string where $\tk_i$ is equal to 1 if $P_k$ on the $i$-th qubit is $\sigma_x$ or $\sigma_y$, and $\tk_i$ is equal to 0 if $P_k$ on the $i$-th qubit is $\sigma_z$ or $\id$. We also have 
\begin{align}
    J^\perp(\psi^{\mathrm{C},(0)}) &:= J\left[\sum_{k:\tk \neq 0^n} \frac{q_k}{\tilde{q}} P_k U_\theta \psi^{\mathrm{C},(0)}  U_\theta^\dagger P_k\right] \nonumber  \\ &\preceq J[\psi^{\mathrm{C},(0)}_\theta]. 
\end{align}
We show in \appref{app:random-stabilizer} that 
\begin{equation}
\label{eq:random-stab}
    J(\psi^{\mathrm{C},(0)}) - \tilde q J^\perp(\psi^{\mathrm{C},(0)}) \succeq \frac{(1-2q)^2-q}{(1-2q)^2} J(\psi^{\mathrm{C},(0)}),
\end{equation}
which indicates that the noisy output state $\mE_\theta(\psi^{\mathrm{C},(0)})$ is a $(1-q,\frac{(1-2q)^2-q}{(1-2q)^2})$-approximately low-rank well-conditioned state. Note that when we use an input state 
$\ket{\psi^{\mathrm{C},(i)}} := U^{(i)}\ket{0^{\otimes n}}$ where $U^{(i)}$ is a Clifford unitary, we have 
\begin{equation}
    \mE_\theta(\psi^{\mathrm{C},(i)}) = U^{(i)} \mE'_\theta (\psi^{\mathrm{C},(0)}) U^{(i)\dagger},
\end{equation}
where $\mE'_\theta$ is equal to $\mE'_\theta$ with all Pauli operators $P_k$, both in the noise channel and the Hamiltonian evolution, replaced by $U^{(i)\dagger} P_k U^{(i)}$. \eqref{eq:random-stab} still holds under such a replacement, and we conclude that $\mE_\theta(\psi^{\mathrm{C},(i)})$ are $(1-q,\frac{(1-2q)^2-q}{(1-2q)^2})$-{approximately low-rank well-conditioned states} for all $i$'s. 
Applying \thmref{thm:unitary}, we know when $M$ is given by a $3$-design,  
\begin{equation}
    I(\psi^{\mathrm{C},(i)},M) \succeq \frac{d+2}{d+1}\frac{(1-q)^2(1-4q)}{2(2-q)(1-2q)^2}  J(\psi^{\mathrm{C},(i)}). 
\end{equation}
Then
\begin{align}
   &\quad I(\{p^{\mathrm{C},(i)},\psi^{\mathrm{C},(i)},\tM\}) = \sum_i p^{\mathrm{C},(i)} I(\psi^{\mathrm{C},(i)},M) 
\nonumber   \\ &\succeq \frac{d+2}{d+1}\frac{(1-q)^2(1-4q)}{2(2-q)(1-2q)^2} \sum_i p^{\mathrm{C},(i)} J(\psi^{\mathrm{C},(i)}) \nonumber \\
   &\succeq \frac{d+2}{d+1}\frac{(1-q)^2(1-4q)}{2(2-q)} \left(\frac{4d}{d+1}\id\right), \label{eq:stab-lower}
\end{align}
where the last line is proved in \appref{app:random-stabilizer}. The deviation observables in this case may not have simple expressions (see discussion in \appref{app:random-stabilizer}). 

To summarize, the discussion above (in particular,\eqref{eq:lower-bell} and \eqref{eq:stab-lower}) shows the near-optimality of the maximally entangled state and the random stabilizer state ensemble when $q$ is sufficiently small. We have the following corollary of \thmref{thm:unitary}. 
\begin{corollary}
    The maximally entangled state, or the random stabilizer state ensemble, together with any $3$-design measurement, forms a near-optimal protocol for the Hamiltonian estimation problem in \eqref{eq:hamt-est} at $\theta = 0$, even under the influence of sufficiently small Pauli noise. 
\end{corollary}

\section{Conclusions and outlook}

In this work, we showed randomized measurements are near-optimal for pure states and {approximately low-rank well-conditioned states}, and are weakly near-optimal for three families of rank-$r$ well-conditioned states. We also provide examples where randomized measurements are explicitly shown to be useful or proved to be suboptimal. 

Our work not only establishes randomized measurements are an important tool in multi-parameter quantum metrology but also highlights several open problems and new research directions. First, although we classified several important families of quantum states for which randomized measurements have the near-optimal performance, it still remains largely open in general which types of mixed states can be estimated with randomized measurements (weakly) near-optimally. Second, we mostly analyze quantum measurements given by $3$-designs and it is desirable to understand the performance of other types of randomized measurements that are even more experimentally friendly, e.g., tensor product of single-qudit $3$-designs (i.e., Pauli measurements in the qubit case), 
or approximate $3$-designs (which can be implemented in log-depth circuits~\cite{schuster2024random}). Third, it will be interesting to study generalization of our results to quantum channel estimation, instead of quantum state estimation, with optimization of input states also taken into consideration. Finally, our discussion is mostly restricted to local parameter estimation, and it is interesting to consider generalization to global parameter estimation, using e.g., Bayesian approach~\cite{gorecki2020pi} and hypothesis testing approach~\cite{meyer2023quantum}.

\section*{Data Availability}
Codes and data for this project are available at the following Github repository: \url{https://github.com/csenrui/random-measurement-metrology}.

\begin{acknowledgments}
We would like to thank Hsin-Yuan Huang for helpful discussions.
S.Z. acknowledges support from Perimeter Institute for Theoretical Physics, a research institute supported in part by the Government of Canada through the Department of Innovation, Science and Economic Development Canada and by the Province of Ontario through the Ministry of Colleges and Universities.
S.C. acknowledges support from ARO (W911NF-23-1-0077), AFOSR MURI (FA9550-21-1-0209), NTT Research, and the Institute for Quantum Information and Matter, an NSF Physics Frontiers Center (NSF Grant PHY-2317110).
S.C. thanks Perimeter Institute for their hospitality where part of this work was completed.
\end{acknowledgments}


\newpage 

\bibliography{refs}

\onecolumngrid
\newpage
\appendix


\setcounter{theorem}{0}
\setcounter{proposition}{0}
\setcounter{lemma}{0}
\setcounter{figure}{0}
\renewcommand{\thefigure}{S\arabic{figure}}
\renewcommand{\thelemma}{S\arabic{lemma}}
\renewcommand{\thetheorem}{S\arabic{theorem}}
\renewcommand{\thecorollary}{S\arabic{corollary}}
\renewcommand{\theproposition}{S\arabic{proposition}}
\renewcommand{\theHfigure}{Supplement.\arabic{figure}}
\renewcommand{\theHlemma}{Supplement.\arabic{lemma}}
\renewcommand{\theHtheorem}{Supplement.\arabic{theorem}}
\renewcommand{\theHcorollary}{Supplement.\arabic{corollary}}

\section{Related works on measurements for multi-parameter metrology}
\label{app:review}

Here we present a detailed comparison of some traditional measurement protocols and our $3$-design measurement protocol, as an extension of \tabref{tab:comparison} in the main text. (Note that we have listed only the most relevant protocols below; there are other useful measurements not included here.)

\begin{table}[h]
    \centering
    \begin{tabular}{c  c  c  c  c c}
            \toprule
              Measurements & Family of States & \makecell[c]{Cost Matrix\\ (Y/N)} & \makecell[c]{Collective \\(Y/N)} & \makecell[c]{~~~Computational~~~\\ Complexity for\\  Identifying POVMs} & \makecell[c]{~~~Implementation~~~\\ Complexity} \\\midrule
              $3$-design & \makecell[c]{Pure, Approximately \\ low-rank, Full-parameter \\
              mixed states (\tabref{tab:summary})} & N  & N  & None  & \makecell[c]{Efficient in \\ multi-qubit systems} \\ 
             \midrule             HCRB~\cite{holevo2011probabilistic,kahn2009local,yamagata2013quantum,yang2019attaining} & Mixed states & Y & \makecell[c]{Y\\ (Infinite copies)} & Unknown & Almost infeasible \\ \midrule
             HCRB~\cite{matsumoto2002new} & Pure states & Y & N & \makecell[c]{${\rm Poly}(d)$, \\ Semidefinite prog-\\
             ramming~\cite{albarelli2019evaluating}} & Unknown \\\midrule
             \makecell[c]{Direct optimization \\ \cite{nagaoka2005new}} &  Mixed states & Y & N & \makecell[c]{NP-hard, \\ Conic prog-\\
             ramming~\cite{hayashi2023tight}} & Unknown \\ \midrule 
             \makecell[c]{Nagaoka--Hayashi \\ \cite{nagaoka2005new,conlon2021efficient}} & \makecell[c]{Qubit states,\\Other examples} & Y & N & - & - \\\midrule
             \makecell[c]{Gill--Massar \\ \citep[Sec.VII]{gill2000state}} & Qubit states & N & N & - & - \\\midrule
             \makecell[c]{Fisher-symmetric \\ \cite{li2016fisher,vargas2024near}} & Pure states & N & N & $\leq {\rm Poly}(d)$ & Unknown \\\midrule             
             Haar~\cite{hayashi1998asymptotic} & \makecell[c]{Pure states}& N & N & None  & Almost infeasible \\\midrule
             $2$-design~\cite{zhu2018universally} & \makecell[c]{Pure states}& N & \makecell[c]{Y \\ (Two copies)} & None  & Platform-dependent \\\midrule
             $2$-design~\cite{zhu2018universally} & \makecell[c]{Full-parameter max-\\ imally mixed states}& N & N & None & \makecell[c]{Efficient in \\ multi-qubit systems} \\
             \bottomrule
        \end{tabular}   
    \caption{Comparison of different types of measurements for multi-parameter metrology. We compare them using multiple features including: (1)~Families of parameterized states, for which the measurements are  optimal (under certain metric). This describes the generality of the measurement protocol. (2)~Dependence on cost matrices, which determines whether the WMSE is minimized based on a specific cost matrix or the CFIM $I(M)$ is optimized to be as close to $J$ as possible. In the former case, the cost matrix must be specified before the measurement is selected and there is no guarantee of optimality under other cost matrices. (3)~Whether collective measurements are required, i.e., whether the POVM needs to be performed on more than one copies of the quantum state. Collective measurements are usually considered to be difficult to implement than individual measurements. (4)~Computational complexity to solve for optimal POVMs. Note that when measurements are state-independent, i.e., universally optimal, e.g., $3$-design, Haar and $2$-design, there is no need to solve for them and thus we mark the computational complexity here as ``None''. The computational complexity required for identifying the optimal POVM for the mixed-state HCRB is expected to be superpolynomial in $d$. Even though the value of HCRB is computable in ${\rm Poly}(d)$~\cite{albarelli2019evaluating}, the known optimal POVM needs to be performed on at least ${\rm Poly}(d)$ copies of quantum states and requires a classical computer of size at least superpolynomial in $d$ to store~\cite{kahn2009local} (For the Nagaoka--Hayashi and Gill--Massar bounds, only special types of states were known to have optimal measurements in explicit forms; and thus we do not specify the computational or implementation complexity in this table.) (5) Implementation complexity. The complexity to implement POVMs was rarely investigated in literature of HCRB (and its variants). The complexity to implement an arbitrary unitary operation is a $d$-dimensional system is $\Omega(d)$~\cite{nielsen2002quantum}. Thus, the implementation complexity of POVMs will be considered exponential in qubit number, unless with known simplifications. The implementation of two-copy measurements, though may be feasible in some physical platforms (e.g.,~\cite{hou2018deterministic}), can suffer from issues in synchronization, cross talks, and decoherence in general and thus is platform-dependent. 
    }
\end{table}

\section{Justification of the definition of weak near-optimality}
\label{app:weak-optimal}

In this appendix, we justify the definition of weak near-optimality (\defref{def:weak}). 

First, we note that it is impossible to define weak near-optimality without choosing a specific cost matrix as in \eqref{eq:def-weak-near-optimal-1}. If for some measurement $M$, there exists a constant $c > 0$ such that for \emph{all} cost matrices $W \succeq 0$, 
\begin{equation}
    \forall \tM, \;\; \trace(W\cdot I(M)^{-1}) \leq c \trace(W\cdot  I(\tM)^{-1}),
\end{equation}
then 
\begin{equation}
\exists\, c>0,\forall \tM, \;\; I(M) \succeq c I(\tM), 
\end{equation}
It then implies 
\begin{equation}
\label{eq:I>J}
    I(M) \succeq c J,
\end{equation}
because if \eqref{eq:I>J} is violated, there must be some specific parameter, which may be a linear combination of $\{\theta_i\}_i$, such that its CFI is smaller than $c$ times its QFI, violating the saturability of single-parameter CR bound (\eqref{eq:single}). 
Therefore, $M$ must be near-optimal. 

Second, we prove some meaningful properties of the weak near-optimality conditions. 
\begin{enumerate}[wide, labelwidth=!,itemindent=!,labelindent=0pt, leftmargin=0em, label={(\arabic*)}, parsep=0pt]
    \item It is independent of parameterization. To be specific, when we view $\rho_\theta$ as $\rho_\varphi$ which is a function of $\varphi$ where $\varphi$ is a new set of parameters defined by $\varphi_i = \sum_{j} A_{ij} \theta_j$ for some invertible 
matrix $A$. We have $A^T I(M) |_{\rho_\varphi} A = I(M)|_{\rho_\theta} $ and $A^T J |_{\rho_\varphi} A = J|_{\rho_\theta}$ and $\trace(J|_{\rho_\varphi} \cdot I(M)^{-1}|_{\rho_\varphi}) = \trace(J|_{\rho_\theta} \cdot I(M)^{-1}|_{\rho_\theta})$, showing \eqref{eq:def-weak-near-optimal-1}'s invariance under reparameterization. \eqref{eq:def-weak-near-optimal-2} is also invariant under reparameterization, noting that $\trace(J^{-1}|_{\rho_\varphi} \cdot I(M)|_{\rho_\varphi}) = \trace(J^{-1}|_{\rho_\theta} \cdot I(M)|_{\rho_\theta})$. 
\item 
It is a necessary condition of near-optimality. If a measurement is near-optimal, then $I(M) \succeq c J \succeq c I(\tM)$ for some $c$ any $\tM$ and \eqref{eq:def-weak-near-optimal-1} is satisfied. \eqref{eq:def-weak-near-optimal-2} also holds because 
\begin{equation}
    I(M) \succeq c \frac{\trace(J^{-1}J)}{\trace(J^{-1}J)} J \succeq c \frac{\trace(J^{-1}I(M))}{m} J. 
\end{equation}
\item It is a proper generalization of the near-optimality condition, i.e., \emph{for any state where near-optimal measurements exist, weakly near-optimal measurements are also near-optimal.} Suppose $\tM$ is near-optimal and $I(\tM) \succeq c J$, and $M$ is weakly near-optimal. Then $\trace(J I(M)^{-1}) \leq c_1 \trace(J I(\tM)^{-1}) \leq \frac{c_1}{c} m$ from \eqref{eq:def-weak-near-optimal-1}, and from \eqref{eq:def-weak-near-optimal-2},
\begin{equation}
    I(M) \succeq c_2 \frac{\trace(J^{-1}I(M))}{m} J \succeq c_2 \frac{m}{m \trace(J I(M)^{-1})} J \succeq \frac{c c_2}{c_1} J. 
\end{equation} 
Therefore, $M$ must be near-optimal. 
\item 
In single-parameter estimation (i.e., $m=1$), the weak near-optimality condition and the near-optimality condition are equivalent. In particular, when $m=1$, \eqref{eq:def-weak-near-optimal-1} is equivalent to the near-optimality condition and \eqref{eq:def-weak-near-optimal-2} is trivially true. Note that \secref{sec:no-go} contains an example of single-parameter estimation on mixed states where randomized measurements are proven to be not near-optimal, and they are also not weakly near-optimal according to this property.
\item It is a proper generalization of the Fisher-symmetry condition. The Fisher-symmetry condition describes the optimality of measurements for local state tomography and is defined only for full-parameter quantum states (i.e., when $m$ is maximal). Specifically, a POVM is called Fisher-symmetric~\cite{li2016fisher,zhu2018universally} if $I(M) \propto J$ and the GM bound is saturated at $\trace(J^{-1} I(M)) = d-1$. Or equivalently, a POVM is Fisher-symmetric if and only if 
\begin{equation}
    I(M) = \frac{d-1}{m} J,
\end{equation} 
which makes sure $M$ provides uniform and maximal information on all parameters.

Below we show 
the Fisher-symmetry condition is \emph{almost}\footnote{Note that in contrast to the Fisher-symmetry condition which defines the optimality of measurement in an exact fashion, the weak near-optimality is defined in an approximate fashion (i.e., we tolerate any constant-factor suboptimality). Therefore, we cannot hope for an exact equivalence between the Fisher-symmetry condition and the weak near-optimality condition. If we fix $c_1=c_2 =1$ in the definition of weak near-optimality, the equivalence will be exact.} equivalent to the weak near-optimality condition. 

We first show \emph{any Fisher-symmetric measurement is weakly near-optimal}. For full-parameter pure states, a measurement is Fisher-symmetric if and only if $I(M) = \frac{1}{2}J$~\cite{li2016fisher} and it is clearly near-optimal and thus weakly near-optimal. 
For full-parameter full-rank mixed states, a measurement is Fisher-symmetric if and only if $I(M) = \frac{1}{d+1}J$. It implies that $\trace(J I(\tM)^{-1})$ is minimized at $m^2/(d-1) = m(d+1)$ when $\tM = M$ from Cauchy--Schwarz (\eqref{eq:lower}). The first equality is satisfied because $I(M) \propto J$ and the second equality holds because $\trace(J^{-1} I(M)) = d-1$ and $m = d^2 - 1$. Thus, $\trace(J I(M)^{-1}) \leq  \trace(J I(\tM)^{-1})$ for any POVM $\tM$. Furthermore, $I(M) = \frac{\trace(J^{-1} I(M))}{m} J$ because $I(M) \propto J$. Thus $M$ is weakly near-optimal. 

Next, we show \emph{whenever ``nearly'' Fisher-symmetric measurements exist, weakly near-optimal measurements are also nearly Fisher-symmetric.} Here a measurement is called  nearly Fisher-symmetric if
\begin{equation}
\label{eq:nearly-f-s}
    I(M) \succeq c \frac{d - 1}{m} J
\end{equation} 
for some positive constant $c$ (as a relaxation of the exact case $I(M) = \frac{d - 1}{m} J$). 
For full-parameter full-rank mixed states, $m = d^2 - 1$. Let $\tM$ be a nearly Fisher-symmetric measurement, \eqref{eq:def-weak-near-optimal-1} implies $\trace(J I(M)^{-1}) \leq c_1 \trace(J I(\tM)^{-1}) \leq \frac{c_1}{c} (d+1)^2(d-1)$. Then 
\begin{equation}
    \trace(I(M) J^{-1}) \geq \frac{m^2}{\trace(J I(M)^{-1})} \geq \frac{c (d^2-1)^2}{c_1 (d+1)(d^2-1)} = \frac{c(d-1)}{c_1}, 
\end{equation}
and then from \eqref{eq:def-weak-near-optimal-2}, 
\begin{equation}
    I(M) \succeq c_2 \frac{\trace(I(M) J^{-1})}{m} J \succeq c_2 \frac{c(d-1)}{c_1 (d^2 -1)} J = \frac{c c_2(d-1)}{c_1 m} J. 
\end{equation}
The same analysis also works for full-parameter pure states. 

Finally, we remark that the discussion above also shows if we require $c_1 = c_2 = 1$ in the definition of weak-optimality (\defref{def:weak}), the Fisher-symmetry condition would be exactly equivalent to the weak near-optimality condition, whenever Fisher-symmetric measurements exist.

\end{enumerate}

\section{Proof of \texorpdfstring{\eqref{eq:K-bound}}{Eq.(59)}}
\label{app:K-bound}

Here we provide a proof of \eqref{eq:K-bound} that appear in \secref{sec:low-rank}.

Let $\vv$ be an arbitrary vector in $\bR^m$, we have 

\begin{align}
    \vv^T \tK \vv &= \sum_{jk:\lambda_j \text{\,or\,} \lambda_k \geq \mu} \left( \frac{2 \bra{\psi_j} \sum_i v_i \partial_i \rho \ket{\psi_k}}{\lambda_j+\lambda_k}  +  \frac{\sum_i v_i \partial_i p}{1-p} \delta_{jk} \right) \left( \frac{2 \bra{\psi_k} \sum_{i'} v_{i'} \partial_i \rho \ket{\psi_j}}{\lambda_j+\lambda_k}  +  \frac{\sum_{i'} v_{i'} \partial_{i'} p}{1-p} \delta_{kj} \right)\\
    &= \sum_{jk:\lambda_j \text{\,or\,} \lambda_k \geq \mu} 4\frac{ \abs{\bra{\psi_j} \sum_i v_i \partial_i \rho \ket{\psi_k}}^2}{(\lambda_j+\lambda_k)^2} + \frac{\sum_i v_i \partial_i p}{1-p} \sum_{j:\lambda_j \geq \mu} 2 \frac{ \bra{\psi_j} \sum_i v_i \partial_i \rho \ket{\psi_j}}{\lambda_j}  + d_\Pi \left(\frac{\sum_i v_i \partial_i p}{1-p}\right)^2 \\
    &= \sum_{jk:\lambda_j \text{\,or\,} \lambda_k \geq \mu} 4\frac{ \abs{\bra{\psi_j} \sum_i v_i \partial_i \rho \ket{\psi_k}}^2}{(\lambda_j+\lambda_k)^2} + 2\frac{\sum_i v_i \partial_i p}{1-p} \left( \sum_{i,j:\lambda_j \geq \mu} v_i \frac{ \partial_i \lambda_j}{\lambda_j} \right) + d_\Pi \left(\frac{\sum_i v_i \partial_i p}{1-p}\right)^2 \\
    &\leq \sum_{jk:\lambda_j \text{\,or\,} \lambda_k \geq \mu} 4\frac{ \abs{\bra{\psi_j} \sum_i v_i \partial_i \rho \ket{\psi_k}}^2}{(\lambda_j+\lambda_k)^2} + \left( \sum_{i,j:\lambda_j \geq \mu} v_i \frac{ \partial_i \lambda_j}{\lambda_j} \right)^2 + (d_\Pi+1) \left(\frac{\sum_i v_i \partial_i p}{1-p}\right)^2\\
    &\leq \sum_{jk:\lambda_j \text{\,or\,} \lambda_k \geq \mu} 4\frac{ \abs{\bra{\psi_j} \sum_i v_i \partial_i \rho \ket{\psi_k}}^2}{(\lambda_j+\lambda_k)^2} + \sum_{jk:\lambda_{j,k}\geq \mu} \frac{\big(\sum_i v_i\partial_i\lambda_j\big)^2}{\lambda_j\lambda_k} + (d_\Pi+1) \left(\frac{\sum_i v_i \partial_i p}{1-p}\right)^2 \\
    &\leq \frac{1}{\mu}\sum_{jk:\lambda_j \text{\,or\,} \lambda_k \geq \mu} 4\frac{ \abs{\bra{\psi_j} \sum_i v_i \partial_i \rho \ket{\psi_k}}^2}{\lambda_j+\lambda_k} + \frac{1}{\mu}\sum_{jk:\lambda_{j,k}\geq \mu} \frac{\big(\sum_i v_i\partial_i\lambda_j\big)^2}{\lambda_j} + (d_\Pi+1) \left(\frac{\sum_i v_i \partial_i p}{1-p}\right)^2\\
    &\leq \frac{2}{\mu} \vv^T  J \vv + \frac{1}{\mu} \vv^T  J \vv + \frac{p(d_\Pi + 1)}{1-p} \vv^T  J^p \vv \\
    &\leq \frac{3}{\mu} \vv^T  J \vv + \frac{d_\Pi + 1}{d_\Pi \mu} \vv^T  J^p \vv \leq \frac{5}{\mu} \vv^T  J \vv. 
\end{align}
where $d_\Pi = \trace(\Pi)$ is the dimension of the subspace $\Pi$. In the third and seventh line we use $\braket{\psi_j|\partial_i\rho|\psi_j} = \partial_i \lambda_j + \lambda_j (\partial_i \braket{\psi_j|\psi_j}) = \partial_i \lambda_j$, in the fourth and fifth line we use $2ab \leq a^2 + b^2$ for any $a,b \in \bR$ and in the eighth line we use $d_\Pi \mu \leq 1-p$, $(d_\Pi + 1)/d_\Pi \leq 2$, and $J^p \preceq J$. \eqref{eq:K-bound} is then proved. 

\section{Generalization of the GM bound}
\label{app:GM}

Here we review the proof of the GM bound~\cite{gill2000state} and extend it to all cases discussed in \thmref{thm:full}. 

To prove the GM bound (\eqref{eq:GM}), it is first important to notice that the GM quantity $\trace(J^{-1}I(M))$ is independent of reparameterization. That means, it is sufficient to prove the GM bound using one specific parameterization of our choice. Let $\rho_\theta = \sum_{k=1}^r \lambda_k \ket{\psi_k}\bra{\psi_k}$ where $\lambda_k > 0$ for $k \in \{1,2,\ldots,r\} =: [r]$. $\lambda_\ell = 0$ for $\ell \in \{r+1,r+2,\ldots,d\}$. (Note that the original proof~\cite{gill2000state} contains the cases of $r = 1,d$ and here we generalize the discussion to arbitrary $r$). We can, without loss of generality, assume the number of unknown parameters is maximal, because increase the number of parameters can only increase the GM quantity $\trace(J^{-1} I(M))$~\cite{gill2000state}. We choose the following complete set of traceless Hermitian operator basis, 
\begin{gather}
\rho_{,k\ell+} = \ket{\psi_k}\bra{\psi_\ell} + \ket{\psi_\ell}\bra{\psi_k},
\quad 
\rho_{,k\ell-} = i\ket{\psi_k}\bra{\psi_\ell} - i\ket{\psi_\ell}\bra{\psi_k},\quad \forall k<\ell,\;k\in[r],\ell\in[d],
\\
\rho_{,b} = \sum_{k=1}^r c_{bk} \ket{\psi_k}\bra{\psi_k},\quad \forall b = 1,\ldots ,r-1,
\end{gather}
where  
\begin{equation}
\sum_{k=1}^r c_{bk} = 0,\quad \sum_{k=1}^r \frac{1}{\lambda_k} c_{b'k}c_{bk} = \delta_{b'b},\quad \forall b = 1,\ldots ,r-1.  
\end{equation}
Further let $c_{rk} = \lambda_k$ for all $k=1,\cdots,r$. We have $C\Lambda^{-1}C^T = \id$ from above, where $C$ is a $r \times r$ real-valued matrix whose entries are given by $c_{bk}$ for $b,k \in [r]$ and $\Lambda^{-1}$ is a diagonal matrix whose diagonal entries are $\lambda_k^{-1}$. It then implies $\Lambda^{-1/2} C^T  C \Lambda^{-1/2} = \id$, i.e.,  
\begin{equation}
\label{eq:prop-c}
\sum_{b=1}^{r} \frac{c_{bk}c_{bk'}}{\sqrt{\lambda_k\lambda_{k'}}} = \sqrt{\lambda_k\lambda_{k'}} + \sum_{b=1}^{r-1} \frac{c_{bk}c_{bk'}}{\sqrt{\lambda_k\lambda_{k'}}} = \delta_{kk'}. 
\end{equation}
Consider a parameterization of $\rho_\theta$ such that $\partial_i \rho_\theta = \rho_{,i}$. The QFIM is 
\begin{equation}
J_{k\ell\pm,k'\ell'\pm'} = \frac{4}{\lambda_k+\lambda_\ell} \delta_{kk'}\delta_{\ell\ell'}\delta_{\pm\pm'},
\quad 
J_{k\ell\pm,b} = 0,\quad 
J_{b,b'} = \delta_{bb'},
\end{equation}
where $(k,\ell) \in [r] \times [d]$ and $b \in [r-1]$.
Since the CFIM $I(M)$ cannot decrease if we replace any measurement operator $M_1$ with its refinement $\{M_2,M_3\}$ where $M_1 = M_2 + M_3$, it is sufficient to consider a POVM with rank-one operators ($\ket{\phi_x}$ is unnormalized),  
\begin{equation}
M_x = \ket{\phi_x}\bra{\phi_x},\quad \ket{\phi_x} = \sum_{k=1}^d a_{x k}\ket{\psi_k}. 
\end{equation}
Then we have 
\begin{align}
\trace(J^{-1} I(M)) &= \sum_{x} \frac{1}{\braket{\phi_x|\rho_\theta|\phi_x}} \left( \sum_{\substack{k<\ell\\k,\ell\in[r]\times[d]}}\sum_{\pm} \frac{\lambda_k+\lambda_\ell}{4} \braket{\phi_x | \rho_{,k\ell\pm} | \phi_x}^2 + \sum_{b=1}^{r-1} \braket{\phi_x|\rho_{,b}|\phi_x}^2 \right)\\
&= \sum_{x} \frac{1}{\braket{\phi_x|\rho_\theta|\phi_x}} \left(\sum_{\substack{k<\ell\\k,\ell\in[r]\times[d]}} (\lambda_k + \lambda_\ell) \abs{a_{x k}}^2\abs{a_{x \ell}}^2 + \sum_{b=1}^{r-1} \bigg(\sum_{k} \abs{a_{x k}}^2 c_{bk}\bigg)^2\right)\\
&= \sum_{x} \frac{1}{\braket{\phi_x|\rho_\theta|\phi_x}} \left(\sum_{k,\ell=1}^r \abs{a_{x k}}^2 \abs{a_{x \ell}}^2  \lambda_k (1 - \lambda_\ell) + \sum_{k=1}^r\sum_{\ell=r+1}^d \abs{a_{x k}}^2 \abs{a_{x \ell}}^2  (\lambda_k + \lambda_\ell) \right)  
\\
&= \sum_{x}  \left(\sum_{\ell=1}^d  \abs{a_{x \ell}}^2  (1 - \lambda_\ell) \right) = \sum_{x} \trace((\id - \rho_\theta) M_x) = d - 1,
\end{align}
where in the first line we simply plug in the definition of $I(M)$ where $I(M)_{ij} = \sum_{x} \frac{\braket{\phi_x|\partial_{i}\rho_\theta|\phi_x}\braket{\phi_x|\partial_{j}\rho_\theta|\phi_x}}{\braket{\phi_x|\rho_\theta|\phi_x}} $, in the third line we use \eqref{eq:prop-c} to eliminate $c_{bk}$, and in the last line we use $\lambda_\ell = 0$ when $\ell \geq r+1$ and $\braket{\phi_x|\rho_\theta|\phi_x} = \sum_{k=1}^r \lambda_k \abs{a_{xk}}^2$.  This proves \eqref{eq:GM} in the main text. 

Consider the second case in \thmref{thm:full}, where $\partial_i \rho_\theta$ are restricted to the support of $\rho_\theta$. In this case, we exclude $\rho_{,k\ell\pm}$ from the derivative operators where $k\in[r]$ and $\ell\in [d]\backslash[r]$. We have 
\begin{align}
\trace(J^{-1} I(M)) &= \sum_{x} \frac{1}{\braket{\phi_x|\rho_\theta|\phi_x}} \left( \sum_{\substack{k<\ell\\k,\ell\in[r]\times[r]}}\sum_{\pm} \frac{\lambda_k+\lambda_\ell}{4} \braket{\phi_x | \rho_{,k\ell\pm} | \phi_x}^2 + \sum_{b=1}^{r-1} \braket{\phi_x|\rho_{,b}|\phi_x}^2 \right)\\
&= \sum_{x} \frac{1}{\braket{\phi_x|\rho_\theta|\phi_x}} \left(\sum_{\substack{k<\ell\\k,\ell\in[r]\times[r]}} (\lambda_k + \lambda_\ell) \abs{a_{x k}}^2\abs{a_{x \ell}}^2 + \sum_{b=1}^{r-1} \bigg(\sum_{k} \abs{a_{x k}}^2 c_{bk}\bigg)^2\right)\\
&= \sum_{x} \frac{1}{\braket{\phi_x|\rho_\theta|\phi_x}} \left(\sum_{k,\ell=1}^r \abs{a_{x k}}^2 \abs{a_{x \ell}}^2  \lambda_k (1 - \lambda_\ell)\right)  
\\
&= \sum_{x}  \left(\sum_{\ell=1}^r  \abs{a_{x \ell}}^2  (1 - \lambda_\ell) \right) = \sum_{x} \trace(\Pi(\id - \rho_\theta)\Pi M_x) = r - 1,
\end{align}
where $\Pi$ is the projector onto the support of $\rho_\theta$. This proves \eqref{eq:lower-2} in the main text.

Consider the third case in \thmref{thm:full}, where $\partial_i \rho_\theta$ are restricted outside the support of $\rho_\theta$. In this case, we exclude $\rho_{,k\ell\pm}$ and $\rho_{,b}$ from the derivative operators where $k,\ell,b\in[r]$. We have 
\begin{align}
\trace(J^{-1} I(M)) &= \sum_{x} \frac{1}{\braket{\phi_x|\rho_\theta|\phi_x}} \left( \sum_{\substack{k,\ell\in[r]\times([d]\backslash[r])}}\sum_{\pm} \frac{\lambda_k+\lambda_\ell}{4} \braket{\phi_x | \rho_{,k\ell\pm} | \phi_x}^2 \right)\\
&= \sum_{x} \frac{1}{\braket{\phi_x|\rho_\theta|\phi_x}} \left(\sum_{k=1}^r\sum_{\ell=r+1}^d \abs{a_{x k}}^2 \abs{a_{x \ell}}^2  (\lambda_k + \lambda_\ell)\right)  
\\
&= \sum_{x}  \left(\sum_{\ell=r+1}^d  \abs{a_{x \ell}}^2  \right) = \sum_{x} \trace((\id - \Pi) M_x) = d - r. 
\end{align}
This proves \eqref{eq:lower-3} in the main text.

\section{Concentration inequality for complex projective Haar measure}\label{app:spherical}

In this section, we prove a concentration inequality of the (complex projective) Haar measure, which is used to derive \eqref{eq:fid_mixed}. 
\begin{lemma} 
Let $\ket{\phi}$ be a fixed $d$-dimensional state. Let $\mu_s$ be the (complex projective) Haar measure. For any $x\in[0,1]$,
    \begin{equation}
        \Pr_{s\sim\mu_s}[\abs{\braket{s|\phi}}^2\le x] \le (d-1)x. 
    \end{equation}
\end{lemma}
\begin{proof}

    Let $\{\ket{0},\ket{1},\cdots,\ket{d-1}\}$ be an orthonormal basis for the $d$-dimensional Hilbert space.
    By the unitary invariance of (complex projective) Haar measure, we can take $\ket{\phi}=\ket{0}$. 
    The state $\ket s$ sampled from the (complex projective) Haar measure can be expressed by (see e.g.~\cite[Sec.~7.2]{watrous2018theory}):
    \begin{equation}
        \ket s := \frac{\ket{\tilde s}}{\sqrt{\braket{\tilde s|\tilde s}}},\quad \ket{\tilde s}:={\sum_{k=0}^{d-1}(X_k+iY_k)\ket{k}}.
    \end{equation}
    where $X_k$, $Y_k$ are independent standard normal random variables. Thus,
    \begin{equation}
    \begin{aligned}
        \abs{\braket{s|0}}^2 = \frac{\abs{\braket{\tilde s|0}}^2}{\braket{\tilde s|\tilde s}} = \frac{X_0^2+Y_0^2}{\sum_{k=0}^{d-1}(X_k^2+Y_k^2)}.
    \end{aligned}
    \end{equation}
    The probability we want to calculate can be computed as 
    \begin{equation}
    \begin{aligned}
        \Pr[\abs{\braket{s|0}}^2\le x] 
        &= \Pr\left[ \sum_{k=1}^{d-1}(X_k^2+Y_k^2)\ge(x^{-1}-1)(X_0^2+Y_0^2) \right]
        \\&= \int_0^{+\infty}\mathrm{d}\alpha\, \frac12 e^{-\alpha/2}\frac{\Gamma(d-1,\frac12(x^{-1}-1)\alpha)}{\Gamma(d-1)}
        \\&= 1 - (1-x)^{d-1}
        \\&\le (d-1)x.
    \end{aligned}
    \end{equation}
    Here we use the fact that $X_0^2+Y_0^2$ and $\sum_{k=1}^{d-1}(X_k^2+Y_k^2)$ are independent $\chi^2$ distributions with $2$ and $2(d-1)$ degrees of freedom, respectively~\cite{weisstein2002chi}.
    The probability density function and the (upper) cumulative distribution function of $\chi^2$ distributions with $k$ degrees of freedom are given by,
    \begin{equation}
        p(\alpha) = \frac{\alpha^{\frac k2-1}e^{-\alpha}}{2^{k/2}\Gamma(k/2)},\quad \int_{\alpha}^{+\infty}\mathrm{d}\beta \, p(\beta) = \frac{\Gamma(k/2,\alpha/2)}{\Gamma(k/2)},
    \end{equation}
    where $\Gamma(\star,\star)$ is the upper incomplete Gamma function and $\Gamma(\star)$ is the Gamma function. 
    This completes the proof.
\end{proof}

\section{Proof of \texorpdfstring{\eqref{eq:noisy-bell}}{Eq.(108)}}
\label{app:bell}

Here we prove \eqref{eq:noisy-bell} in the main text, i.e., 
\begin{equation}
    J(\psi^{\rm ME})_{ij} = J\left[(\mE_\theta\otimes \id)(\psi^{\rm ME})\right]_{ij} = 2\delta_{ij} \sum_{k:q_k + q_{k+i}>0} \frac{\abs{q_{k+i} - q_k}^2}{q_k + q_{k+i}}.
\end{equation}
First, we note that the eigenstates of $\rho_\theta: = (\mE_\theta\otimes \id)(\psi^{\rm ME})$ are 
\begin{equation}
    \ket{\psi_k} := P_k \otimes \id \ket{\psi^{\rm ME}}, 
\end{equation}
with eigenvalues $q_k$ for $k=0,\ldots,d^2-1$. 
Then
\begin{equation}
    J(\psi^{\rm ME})_{ij} = 2 \sum_{k\ell:q_k+q_\ell>0} \frac{\Re[\braket{\psi_k|\partial_i\rho_\theta| \psi_\ell}\braket{\psi_\ell|\partial_j\rho_\theta|\psi_k}]}{q_k + q_\ell}. 
\end{equation}
Note that 
\begin{align}
    \bra{\psi_k}\partial_i\rho_\theta\ket{\psi_\ell} 
    &= - \bra{\psi_k} \left(\sum_{k'} i q_{k'} P_{k'} P_i \ket{\psi_0}\bra{\psi_0}P_{k'} - i q_{k'} P_{k'}  \ket{\psi_0}\bra{\psi_0} P_i P_{k'} \right) \ket{\psi_\ell} \nonumber\\
    &= \sum_{k'} - i q_{k'} \eta_{k'i} \delta_{k,k'+i}\delta_{k'\ell} + i q_{k'} \eta_{k'i}^* \delta_{kk'} \delta_{i+k',\ell} = - \delta_{k+i,\ell} i q_\ell \eta_{\ell i} + \delta_{k+i,\ell} i q_k \eta_{ki}^*. 
\end{align}
Here $\eta_{ij}$ is defined through $P_{i+j} = \eta_{ij} P_iP_j$ and $\eta_{k,i} \in \{\pm 1,\pm i\}$ makes sure $P_{k+i} \in \{\id,\sigma_x,\sigma_y,\sigma_z\}^{\otimes n}$. $\eta_{ij} = \pm 1$ when $[P_i,P_j] = 0$ and $\eta_{ij} = \pm i$ when $[P_i,P_j] \neq 0$. We also use $\braket{\psi_0|P_iP_j|\psi_0} = \frac{1}{d}\trace(P_iP_j) = \delta_{ij}$ above. Then 
\begin{align}
    \Re[\braket{\psi_k|\partial_i\rho_\theta| \psi_\ell}\braket{\psi_\ell|\partial_j\rho_\theta|\psi_k}]
    &= \delta_{ij} \delta_{k+i,\ell} \abs{ q_\ell \eta_{\ell i} - q_k \eta_{ki}^*}^2.
\end{align}
We have 
\begin{equation}
    J(\psi^{\rm ME})_{ij} = 2 \delta_{ij}\sum_{k\ell:q_k+q_\ell>0} \frac{\delta_{k+i,\ell} \abs{ q_\ell \eta_{\ell i} - q_k \eta_{ki}^*}^2}{q_k + q_\ell} =  2 \delta_{ij} \sum_{k:q_k+q_{k+i}>0} \frac{\abs{ q_{k+i} \eta_{k+i,i} - q_k \eta_{ki}^*}^2}{q_k + q_{k+i}}. 
\end{equation}
Finally, note that 
\begin{equation}
    P_{k+i} P_k = (P_{k+i}P_i)(P_i P_k) = (\eta_{k+i,i} P_k)(\eta_{ik} P_{k+i}) = \eta_{k+i,i}\eta_{ik} P_k P_{k+i}. 
\end{equation}
When $[P_i,P_k] = 0$, we have $[P_{k+i},P_k] = 0$ and $\eta_{ik} = \pm 1 = - \eta_{k+i,i}$. When $[P_i,P_k] \neq 0$, we have $[P_{k+i},P_k] \neq 0$ and $\eta_{ik} = \pm i = \eta_{k+i,i}$. Also note that $\eta_{ik} = \eta_{ki}^*$. We have 
\begin{equation}
    J(\psi^{\rm ME})_{ij} =  2 \delta_{ij}\sum_{k:q_k+q_{k+i}>0} \frac{\abs{ q_{k+i} - q_k}^2}{q_k + q_{k+i}}. 
\end{equation}

\section{Proofs of \texorpdfstring{\eqref{eq:random-stab}}{Eq.(114)} and \texorpdfstring{\eqref{eq:stab-lower}}{Eq.(117)}, Computing deviation observables}
\label{app:random-stabilizer}

Here we first prove \eqref{eq:random-stab} in the main text, i.e., 
\begin{equation}
    J(\psi^{\mathrm{C},(0)}) - \tilde q J^\perp(\psi^{\mathrm{C},(0)}) \succeq \frac{(1-2q)^2-q}{(1-2q)^2} J(\psi^{\mathrm{C},(0)})
\end{equation}

We first compute $J(\psi^{\mathrm{C},(0)})$, given by
\begin{align}
    J(\psi^{\mathrm{C},(0)})_{jj'} &= J\left[\mE_\theta(\psi^{\mathrm{C},(0)})\right] = J\left[\sum_{k=0}^{d^2-1} q_k P_k U_\theta \psi^{\mathrm{C},(0)} U_\theta^\dagger P_k\right]_{jj'} \nonumber\\
    &= \sum_{\substack{x,y\in\{0,1\}^{n}\\ \tq_{x} + \tq_{y}>0}} \frac{2\Re[\bra{x}\partial_j\rho_\theta\ket{y}\bra{y}\partial_{j'}\rho_\theta\ket{x}]}{\tq_{x} + \tq_{y}},
\label{eq:noisy-random}
\end{align}
where $\rho_\theta := \mE_\theta(\psi^{\mathrm{C},(0)})$ and $\ket{x} = \ket{x_1,x_2,\ldots,x_n}$ are the computational basis, as eigenstates of $\rho_\theta$ at $\theta = 0$. $\tq_x := \sum_{k:\tk = \tilde{x}} q_k$ are the corresponding eigenvalues. 
Note that
\begin{equation}
\bra{x}\partial_j\rho_\theta\ket{y} =\sum_k q_k \left( -i\bra{x}P_kP_j\ket{0^{\otimes n}}\bra{0^{\otimes n}}P_k\ket{y}+i\bra{x}P_k\ket{0^{\otimes n}}\bra{0^{\otimes n}}P_jP_k\ket{y} \right)    
\end{equation}
is non-zero only when $\tj = \tilde{x} + \tilde{y}$. This implies $J(\psi^{\mathrm{C},(0)})$ must be block-diagonal in the sense that $J(\psi^{\mathrm{C},(0)})_{jj'} = J(\psi^{\mathrm{C},(0)})_{jj'} \delta_{\tj\tj'}$. 
Below we will use a matrix $B$ to lower bound $J(\psi^{\mathrm{C},(0)})$ where $J(\psi^{\mathrm{C},(0)}) \succeq B$, and 
\begin{equation}
B_{jj'} := \delta_{\tj\tj'} \left( \frac{2\Re[\bra{0^{\otimes n}}\partial_j\rho_\theta\ket{\tj}\bra{\tj}\partial_{j'}\rho_\theta\ket{0^{\otimes n}}]}{\tq_{0} + \tq_{\tj}} + \frac{2\Re[\bra{\tj}\partial_j\rho_\theta\ket{0^{\otimes n}}\bra{0^{\otimes n}}\partial_{j'}\rho_\theta\ket{\tj}]}{\tq_{0} + \tq_{\tj}} \right),    
\label{eq:B}
\end{equation}
where we specifically pick out two terms $(x,y) = (0,\tj)$ and $(x,y) = (\tj,0)$ from $J(\psi^{\mathrm{C},(0)})$. On the other hand, we have 
\begin{equation}
    J^\perp(\psi^{\mathrm{C},(0)}) \preceq J[\psi^{\mathrm{C},(0)}_\theta],\quad 
    J[\psi^{\mathrm{C},(0)}_\theta]_{jj'} = 4 ( \bra{0^{\otimes n}}\!\frac{1}{2}\{P_j,P_{j'}\}\!\ket{0^{\otimes n}} - \bra{0^{\otimes n}}\!P_j\!\ket{0^{\otimes n}}\bra{0^{\otimes n}}\!P_{j'}\!\ket{0^{\otimes n}} ).
\end{equation}

Our goal is then to show 
\begin{equation}
\label{eq:stab-proof}
    B - \tilde q J[\psi^{\mathrm{C},(0)}_\theta] \succeq \frac{(1-2q)^2-q}{(1-2q)^2} B. 
\end{equation}
Both $B_{jj'}$ and $J[\psi^{\mathrm{C},(0)}_\theta]_{jj'}$ are zero when $\tj \neq \tj'$. Therefore, it is sufficient to show \eqref{eq:stab-proof} for blocks of matrices with the same $\tj$, i.e,
\begin{equation}
\label{eq:stab-proof-1}
    \projP_{\tj} B \projP_{\tj} - \tilde q \projP_{\tj} J[\psi^{\mathrm{C},(0)}_\theta] \projP_{\tj} \succeq \frac{(1-2q)^2-q}{(1-2q)^2} \projP_{\tj} B \projP_{\tj}. 
\end{equation}
for all $\tj \in \{0,1\}^n$, where the projector $\projP_{\tj} = \sum_{j':\tj'=\tj} \ket{j'}\bra{j'}$. 
Note that when $\projP_0 B \projP_0 = \projP_0 J[\psi^{\mathrm{C},(0)}_\theta] \projP_0 = 0$. Therefore, we assume $\tj \neq 0^n$, and prove \eqref{eq:stab-proof-1} below. We have  
\begin{align}
\bra{0^{\otimes n}}\partial_j\rho_\theta\ket{\tj} 
&=\sum_{k,\tk=\tj} q_k \left( -i\bra{0^{\otimes n}}P_kP_j\ket{0^{\otimes n}}\bra{0^{\otimes n}}P_k\ket{\tj}\right) + \sum_{k,\tk=0} q_k \left(i\bra{0^{\otimes n}}P_k\ket{0^{\otimes n}}\bra{0^{\otimes n}}P_jP_k\ket{\tj} \right)\nonumber\\
&= \sum_{k,\tk=\tj} q_k \left( -i (-1)^{\nu_k}i^{\nu_j} \right) + \sum_{k,\tk=0} q_k \left(i(-i)^{\nu_j}(-1)^{z_k\cdot\tj} \right) \nonumber \\
&= q_0 i(-i)^{\nu_j} \left(  1 + (-1)^{\nu_j+1} \sum_{k,\tk=\tj} \frac{q_k}{q_0}   (-1)^{\nu_k}  + \sum_{k:k\neq0,\tk=0} \frac{q_k}{q_0} (-1)^{z_k\cdot\tj}  \right) \nonumber\\
&= q_0 i \braket{0^{\otimes n}|P_j|\tj} \left(  1 + (-1)^{\nu_j+1} \sum_{k,\tk=\tj} \frac{q_k}{q_0}   (-1)^{\nu_k}  + \sum_{k:k\neq0,\tk=0} \frac{q_k}{q_0} (-1)^{z_k\cdot\tj}  \right)\nonumber\\
&=: q_0 i\braket{0^{\otimes n}|P_j|\tj} \left(  1 + (-1)^{\nu_j+1} g_1(\tj) + g_2(\tj) \right), \text{~~~where~}i\braket{0^{\otimes n}|P_j|\tj} = (-i)^{\nu_j}.
\label{eq:coeff}
\end{align}
Here $\nu_k$ denotes the number of Pauli Y's in $P_k$ and $z_k$ is an $n$-bit string whose elements are $1$ at locations where $P_k$ is Pauli Z and $0$ otherwise. $\abs{g_{1}(\tj)} + \abs{g_{2}(\tj)} \leq \frac{q}{q_0}$. Then 
\begin{align}
    B_{jj'} &= \delta_{\tj\tj'} \frac{4q_0^2 \Re[\braket{0^{\otimes n}|P_j|\tj}\braket{\tj|P_{j'}|0^{\otimes n}}] \left(  1 + (-1)^{\nu_j+1} g_1(\tj) + g_2(\tj) \right)  \left(  1 + (-1)^{\nu_{j'}+1} g_1(\tj) + g_2(\tj) \right)}{\tq_{0} + \tq_{\tj}}  \label{eq:exp-B} \\
    &= J[\psi^{\mathrm{C},(0)}_\theta]_{jj'} \delta_{\tj\tj'}  \frac{q_0^2\left(  1 + (-1)^{\nu_j+1} g_1(\tj) + g_2(\tj) \right)  \left(  1 + (-1)^{\nu_{j'}+1} g_1(\tj) + g_2(\tj) \right)}{\tq_{0} + \tq_{\tj}} 
\end{align}
When $\nu_j \neq \nu_{j'}$, the numbers of Pauli Y's in $P_j$ and $P_{j'}$ do not have the same parity. Plus that $\tilde j =\tilde {j'}$, we have $\{P_j,P_{j'}\} = 0$ and $B_{jj'} = J[\psi^{\mathrm{C},(0)}_\theta]_{jj'} = 0$. Therefore, we have 
\begin{align}
    B_{jj'} &= J[\psi^{\mathrm{C},(0)}_\theta]_{jj'} \delta_{\tj\tj'}\delta_{\nu_j\nu_{j'}} \frac{q_0^2\left(  1 + (-1)^{\nu_j+1} g_1(\tj) + g_2(\tj) \right)^2}{\tq_{0} + \tq_{\tj}} \nonumber\\ 
    &\geq 
    J[\psi^{\mathrm{C},(0)}_\theta]_{jj'} \delta_{\tj\tj'}\delta_{\nu_j\nu_{j'}} \frac{q_0^2 (1 - q/q_0)^2}{\tq_{0} + \tq_{\tj}} \geq J[\psi^{\mathrm{C},(0)}_\theta]_{jj'} \delta_{\tj\tj'}\delta_{\nu_j\nu_{j'}} (1-2q)^2.
\end{align}
Therefore,
\begin{equation}
    B \succeq J[\psi^{\mathrm{C},(0)}_\theta] (1-2q)^2. 
\end{equation}
This implies \eqref{eq:stab-proof}, noting that $\tq \leq q$. 

\medskip

The above implies
\begin{equation}
    J(\psi^{\mathrm{C},(0)}) \succeq J[\psi^{\mathrm{C},(0)}_\theta] (1-2q)^2,
\end{equation}
and this extends to all random stabilizer input states as explained in the main text. Then 
\begin{equation}
    \sum_i p^{\mathrm{C},(i)} J(\psi^{\mathrm{C},(i)}) \succeq \sum_i p^{\mathrm{C},(i)} J[\psi^{\mathrm{C},(i)}_\theta] (1-2q)^2 = \frac{4d}{d+1}(1-2q)^2 \id,
\end{equation}
which proves the last line in \eqref{eq:stab-lower}, where we use 
\begin{multline}
    \sum_i p^{\mathrm{C},(i)} J[\psi^{\mathrm{C},(i)}_\theta]_{jj'} = \\ 4 \sum_i p^{\mathrm{C},(i)} ( \bra{0^{\otimes n}}\!U^{(i)\dagger}\frac{1}{2}\{P_j,P_{j'}\}U^{(i)}\!\ket{0^{\otimes n}} - \bra{0^{\otimes n}}\!U^{(i)\dagger}P_jU^{(i)}\!\ket{0^{\otimes n}}\bra{0^{\otimes n}}\!U^{(i)\dagger}P_{j'}U^{(i)}\!\ket{0^{\otimes n}} )  = \frac{4d}{d+1} \delta_{jj'}. 
\end{multline}
Here we use the $2$-design property of the Clifford group. The same derivation has appeared in~\eqref{eq:use_two_design}.

\medskip
Finally, we explain how to compute the deviation observables. They are 
\begin{align}
\label{eq:dev-ob-2}
X_j(\psi^{\mathrm{C},(i)};\theta) 
= \sum_{k} (\tJ(\rho^{\mathrm{C}}_{\rm{RP}})^{-1})_{jk} \tL_k(\psi^{\mathrm{C},(i)}).
\end{align}
When $\psi^{\mathrm{C},(i)} = \ket{0^{\otimes n}}\bra{0^{\otimes n}}$, 
\begin{gather}
\label{eq:tJ-0}
\tJ(\psi^{\mathrm{C},(0)})_{jj'} = B_{jj'} \text{~~~(details in \eqref{eq:B}, \eqref{eq:coeff})},\\
\label{eq:tL-0}
\tL_j(\psi^{\mathrm{C},(0)}) = \frac{2}{\tq_{0} + \tq_{j}}  \left( {\bra{0^{\otimes n}}\partial_j\rho_\theta\ket{\tj}}\ket{0^{\otimes n}}\bra{\tj} + {\bra{\tj}\partial_{j}\rho_\theta\ket{0^{\otimes n}}}\ket{\tj}\bra{0^{\otimes n}} \right) \text{~~~(details in \eqref{eq:coeff})},
\end{gather}
Generally, when $\ket{\psi^{\mathrm{C},(i)}} = U^{(i)\dagger}\ket{0^{\otimes n}}$, $\tJ(\psi^{\mathrm{C},(i)})_{jj'}$ and $\tL_j(\psi^{\mathrm{C},(i)})$ can be calculated again using \eqref{eq:tJ-0} and \eqref{eq:tL-0} with the following transformations:
\begin{enumerate}[wide, labelwidth=!,itemindent=!,labelindent=0pt, leftmargin=0em, label={(\arabic*)}, parsep=0pt]
    \item $\tk$ is an $n$-bit string where the $l$-th number $\tk_l$ is equal to 1 if $U^{(i)\dagger} P_k U^{(i)}$ on the $l$-th qubit is $\sigma_x$ or $\sigma_y$, and $\tk_l$ is equal to 0 if $U^{(i)\dagger} P_k U^{(i)}$ on the $l$-th qubit is $\sigma_z$ or $\id$ and $\tq_x = \sum_{k:\tk = \tilde{x}} q_k$ accordingly. 
    \item Similarly in \eqref{eq:coeff}, $\nu_k$ denotes the number of Pauli Y's in $U^{(i)\dagger} P_k U^{(i)}$ and $z_k$ is an $n$-bit string whose elements are $1$ at locations where $U^{(i)\dagger} P_k U^{(i)}$ is Pauli Z and $0$ otherwise.
    \item $\ket{0^{\otimes n}}\bra{\tj}$ and $\ket{\tj}\bra{0^{\otimes n}}$ in \eqref{eq:tL-0} should be replaced with $U^{(i)}\ket{0^{\otimes n}}\bra{\tj}U^{(i)\dagger}$ and $U^{(i)}\ket{\tj}\bra{0^{\otimes n}}U^{(i)\dagger}$. 
\end{enumerate}

Then deviation observables (\eqref{eq:dev-ob-2}) can be computed using \eqref{eq:tJ-0}, \eqref{eq:tL-0} and $\tJ(\rho^{\mathrm{C}}_{\rm{RP}}) = \bE_i[\tJ(\psi^{\mathrm{C},(i)})]$. 

Let us consider a special case of global depolarizing noise, where $q_0 = 1-q$ and $q_k = \frac{q}{d^2 - 1}$ for all $k \neq 0$. When the input state is $\ket{0^{\otimes n}}$, we have for all $\tj \neq 0$, 
\begin{gather}
    g_1(\tj) = \sum_{k,\tk=\tj} \frac{q_k}{q_0}   (-1)^{\nu_k} = 0 ,\quad 
    g_2(\tj) = \sum_{k:k\neq0,\tk=0} \frac{q_k}{q_0} (-1)^{z_k\cdot\tj} = -\frac{q}{(d^2 -1)(1-q)},\\
    \bra{0^{\otimes n}}\partial_j\rho_\theta\ket{\tj}  = q_0 i (-i)^{\nu_j} \left(  1 + (-1)^{\nu_j+1} g_1(\tj)  + g_2(\tj)  \right) = i (-i)^{\nu_j} \left(1-q + \frac{q}{d^2 -1}\right),\\
\tJ(\psi^{\mathrm{C},(0)})_{jj'} = B_{jj'} = 4\delta_{\tj\tj'}\delta_{\nu_j+\nu_{j'}\text{ is even}}\frac{\left(1-q + \frac{q}{d^2 -1}\right)^2}{1-q + \frac{q(2d - 1)}{d^2 - 1}} (-1)^{\nu_j} (-1)^{\frac{\nu_j+\nu_{j'}}{2}},    \\
~~\Rightarrow~~  \tJ(\psi^{\mathrm{C},(0)})_{jj'} = \frac{\left(1-q + \frac{q}{d^2 -1}\right)^2}{1-q + \frac{q(2d - 1)}{d^2 - 1}} J[\psi^{\mathrm{C},(0)}_\theta]_{jj'}, \label{eq:G22}\\
\tL_j(\psi^{\mathrm{C},(0)}) = \frac{1-q + \frac{q}{d^2 -1}}{1 - q + \frac{q(2d-1)}{d^2-1}}  L_j[\psi^{\mathrm{C},(0)}_\theta]. \label{eq:G23}
\end{gather}
\eqref{eq:G22} and \eqref{eq:G23} show $\tJ$ and $\tL$ are proportional to their noiseless versions. The same holds when input state is any stabilizer state $U^{\mathrm{C},(i)}\ket{0^{\otimes n}}$. Therefore, for global depolarizing noise, we have 
\begin{align}
X_j(\psi^{\mathrm{C},(i)};\theta) 
= \frac{i(d+1)}{2d(1-q + \frac{q}{d^2-1})} \big[\ket{\psi^{\mathrm{C},(i)}_\theta}\!\bra{\psi^{\mathrm{C},(i)}_\theta}, P_j\big],
\end{align}
which in the noiseless case (i.e., $q = 0$) equal to \eqref{eq:X-C}. 
In general, it is difficult to obtain closed-form expressions of deviation observables for random stabilizer states under Pauli noise because $\tJ(\rho^{\mathrm{C}}_{\rm{RP}})$ may not be diagonal. 
\end{document}